\documentclass[twocolumn,
pra,aps,superscriptaddress,showpacs,amsmath,amssymb]{revtex4-2}

\usepackage{graphicx}
\usepackage{dcolumn}
\usepackage{bm}

\usepackage[utf8]{inputenc}
\usepackage[english]{babel}
\usepackage[T1]{fontenc}
\usepackage{amsmath,amssymb,bm,amsthm,amsfonts}

\usepackage{hyperref}

\usepackage{tikz}
\usepackage{lipsum}
\usepackage{tabularx}
\usepackage{subfig}
\usepackage{graphicx}
\usepackage{floatrow}
\usepackage{quantikz}
\usepackage{natbib}
\usepackage{ragged2e}

\usepackage[toc,page]{appendix}

\usepackage[normalem]{ulem}

\usepackage[
ruled, 
lined, 
linesnumbered, 
commentsnumbered, 
longend]{algorithm2e}

\usepackage{xcolor}

\SetCommentSty{mycommfont}

\newtheorem{theorem}{Theorem}
\newtheorem{lemma}{Lemma}

\newtheorem{definition}{Definition}

\usepackage{mathtools}

\usepackage{color}

\begin{document}

\preprint{APS/123-QED}

\title{Universal Analog Quantum Simulation}

 \author{Yiming Huang}
 \affiliation{Center on Frontiers of Computing Studies, Peking University, Beijing 100871, China}
 \affiliation{Institute of High Energy Physics, Chinese Academy of Sciences, Beijing 100049, China}
\affiliation{China Center of Advanced Science and Technology, Beijing 100190, China}

 \author{Jiaxing Song}
 \affiliation{School of Computer Science, Peking University, Beijing 100871, China}
\affiliation{Center on Frontiers of Computing Studies, Peking University, Beijing 100871, China}
\author{Xiaoxia Cai}
\email{xxcai@ihep.ac.cn}
\affiliation{Institute of High Energy Physics, Chinese Academy of Sciences, Beijing 100049, China}
\author{Xiao Yuan}
\email{xiaoyuan@pku.edu.cn}
\affiliation{Center on Frontiers of Computing Studies, Peking University, Beijing 100871, China}
\affiliation{School of Computer Science, Peking University, Beijing 100871, China}

\begin{abstract}

Analog quantum simulators emulate complex many-body dynamics through native continuous-time evolution under hardware-defined interactions. Yet once a platform is specified, its interaction structure is largely fixed by the underlying hardware, restricting the Hamiltonians that can be realized and limiting programmability. Here we introduce universal analog quantum simulation (UAQS), a hybrid framework that systematically expands the range of accessible quantum evolutions within a given analog platform. UAQS employs optimized continuous-time control fields to engineer target dynamics directly, avoiding decomposition into discrete gate sequences. By preserving native analog evolution while extending the set of achievable Hamiltonians, UAQS transforms fixed-interaction analog devices into programmable simulators. Numerical studies on representative architectures, including superconducting circuits and Rydberg-atom arrays, show that UAQS accurately reproduces non-trivial many-body dynamics beyond the intrinsic interaction structure of the hardware. These results establish UAQS as a practical route toward programmable analog quantum simulation.
\end{abstract}
\maketitle

\section{Introduction}

Accurately predicting ground and excited states, as well as the dynamical behavior of quantum many-body systems, is central to a wide range of problems in condensed matter physics, quantum chemistry, and high-energy physics~\cite{feynman2018simulating,georgescu2014quantum,mcardle2020quantum,di2024quantum,miessen2023quantum}. However, classical numerical methods are fundamentally limited by the exponential growth of Hilbert space, the complex entanglement structure of generic many-body states, and the breakdown of perturbative descriptions in strongly correlated regimes, rendering many problems computationally intractable. These limitations motivate the development of quantum simulation as a means to access regimes beyond classical reach.

Among the available approaches, analog quantum simulation offers a direct and high-fidelity route to emulating many-body dynamics by exploiting the native continuous-time interactions of engineered quantum platforms such as cold atoms~\cite{gross2017quantum}, trapped ions~\cite{wu2016understanding,blatt2012quantum}, Rydberg arrays~\cite{browaeys2020many}, photonic~\cite{aspuru2012photonic}, superconducting circuits~\cite{daley2022practical,monroe2021programmable,houck2012chip,altman2021quantum}. By avoiding digital gate decomposition~\cite{lloyd1996universal,lanyon2011universal,kamakari2022digital,fischer2026enabling}, it reduces the running time and operational overhead. However, its scope is inherently constrained by the underlying hardware architecture. For a given hardware platform, the accessible interaction structure and control Hamiltonians are largely determined by the underlying microscopic implementation. Consequently, an analog device engineered for one interaction pattern generally lacks the flexibility to reproduce the dynamics of qualitatively different Hamiltonians which limits its programmability.

Recent efforts have sought to enhance flexibility through advanced control techniques and hybrid digital-analog protocols~\cite{eckardt2017colloquium,babukhin2020hybrid,gonzalez2021ditital,celeri2023digital,garcia2025hamiltonian,kumar2025digital,fauseweh2024quantum}. While these approaches expand the accessible dynamics to varying degrees, they do not provide a general and systematic framework for extending the dynamical and state-preparation capabilities of a fixed analog platform while preserving native continuous-time evolution. Consequently, the ability to transform a hardware-specific analog device into a broadly programmable dynamical resource remains limited.

In this work, we introduce universal analog quantum simulation (UAQS), a general framework that embeds optimization directly into native analog quantum control. UAQS formulates both real-time and imaginary-time evolution as continuous-time control optimization problems executed within the hardware itself. Instead of compiling target Hamiltonians into sequences of discrete gates, UAQS optimizes time-dependent control fields to steer the system toward desired quantum states or dynamical trajectories. By integrating the optimization paradigm of digital quantum algorithms with the intrinsic controllability of analog platforms, the framework enables systematic parametrized control of Hamiltonian dynamics while preserving continuous-time evolution.

UAQS extends the range of many-body dynamics and quantum states accessible to a fixed analog platform beyond its intrinsic interaction structure while preserving the hardware efficiency and low-depth advantages of analog simulation. Through general theoretical analysis and numerical studies on representative architectures, including superconducting circuits and Rydberg atom arrays, we demonstrate that UAQS enables accurate realization of nontrivial dynamics that would otherwise lie outside the native interaction pattern of the hardware. By bridging parametrized analog control and conventional variational principle, UAQS establishes a flexible and experimentally viable paradigm for programmable analog quantum simulation on near-term devices.

\begin{figure*}[ht]
\centering
  \includegraphics[width=\textwidth]{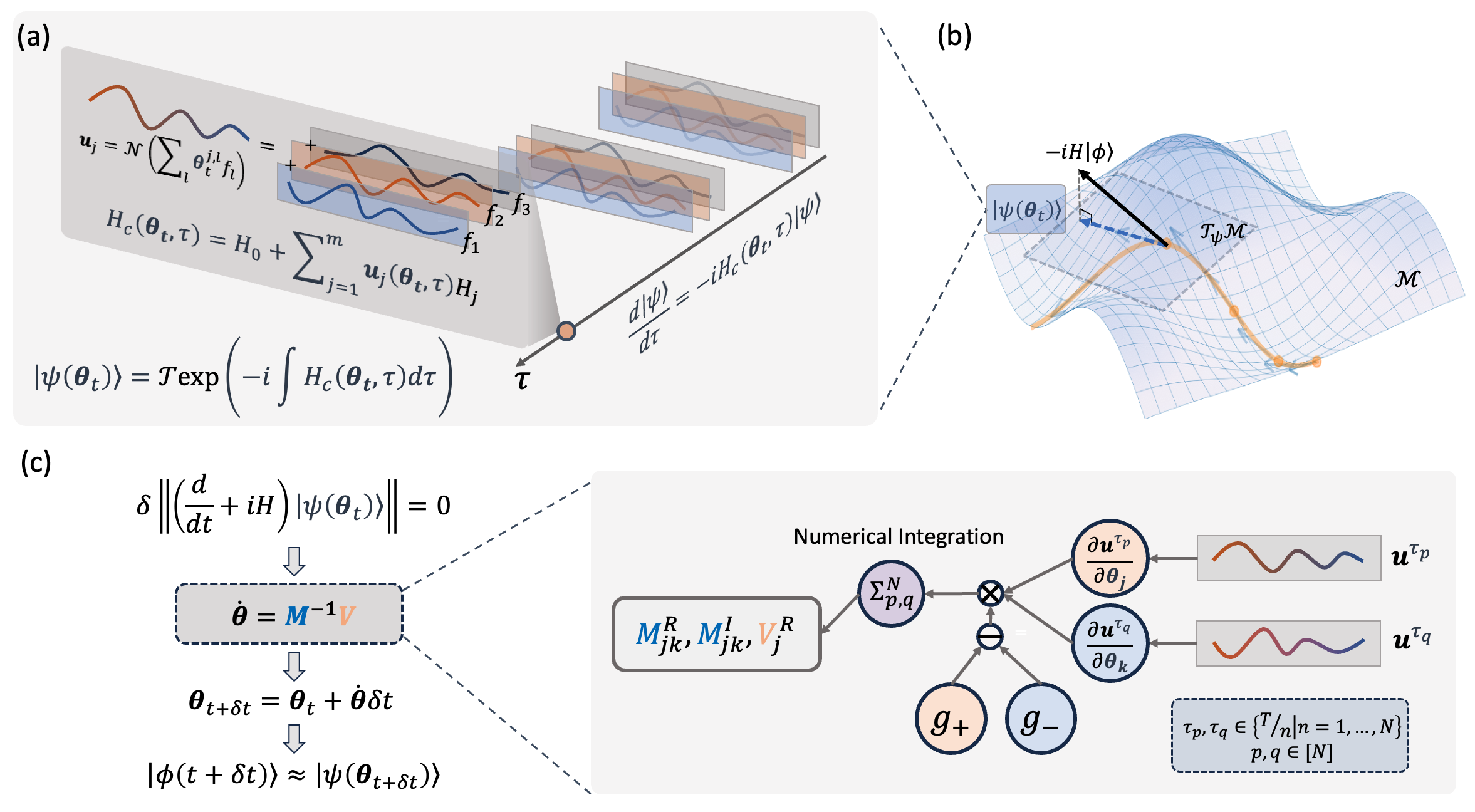}
  \caption{An Illustration of universal analog quantum simulation. (a) The construction of a parametrized analog Hamiltonian driven by a control Hamiltonian $H_c$ with programmable pulses $\bm{u}$. (b) The trial state traces a trajectory on a manifold $\mathcal{M}$ embedded in Hilbert space. The dynamics follow the projection of the exact time evolution $-iH|\phi\rangle$, presented by a black solid arrow, onto the tangent space $T_{\psi}\mathcal{M}$ at $|\psi(\bm{\theta}_t)\rangle$ as shown by a blue dashed arrow. (c) The left panel shows the workflow for deriving the evolution equation of the control parameters $\bm{\theta}$ based on variational principles, while the right panel presents a general framework for estimating the elements of the the associated matrices $M$ and $V$.
}
  \label{fig:framework}
\end{figure*}
\section{Results} 
This section introduces the theoretical ingredients underlying UAQS. We first describe the parametrized analog Hamiltonian that characterizes the dynamical capabilities of an analog platform, and then discuss the UAQS framework for simulating the imaginary-time and real-time evolution.

\subsection{Preliminaries}
A fundamental challenge in quantum simulation arises from the structural mismatch between the dynamics one seeks to emulate and the physical interactions natively available on a given hardware platform. Conventional digital approaches resolve this tension by compiling target Hamiltonians into discrete sequences of elementary quantum gates~\cite{li2017efficient,cirstoiu2020variational,yao2021adaptive,heya2023subspace}. 
While versatile, this compilation paradigm inevitably introduces overhead, accumulates approximation errors through Trotterization, and a disconnect from the continuous-time physics that analog platforms are naturally suited to capture. 

In conventional analog simulation, the control Hamiltonian $H_c$ is typically required to coincide with the target Hamiltonian $H$, which significantly limits the range of problems that can be addressed. To enable the simulation of more general target Hamiltonians $H$, it is necessary to introduce additional flexibility into $H_c$. In this work, we therefore consider a parametrized form of the control Hamiltonian $H_c$ given by

\begin{equation}\label{Eq:analogH}
H_c(\bm{\theta},\tau)=H_0+\sum_{j=1}^m u_j(\bm{\theta}_j,\tau)H_j,
\end{equation}
where $H_0$ denotes a fixed system Hamiltonian and $\{H_j\}$ are control Hamiltonians driven by time-dependent pulses. Each control pulse is described by a shape function $u_j(\bm{\theta}_j,\tau)$ and parametrized as
\begin{equation}
u_j(\bm{\theta}_j,\tau)=\mathcal{N}\left(\sum_{l=1}^d \theta_{jl} f_l(\tau)\right),
\end{equation}
with $\{f_l\}_{l=1}^d$ denoting a set of basis functions and $\mathcal{N}$ a normalization function. Each control Hamiltonian $H_j$ is thus independently governed by its own set of tunable parameters $\bm{\theta}_j$. Because of physical constraints, pulse amplitudes must remain bounded, therefore it imposes the constraint by introducing a normalizer $\mathcal{N}$ ensuring each pulse is bounded by $\theta_{\max}$ which is chosen according to the target device.

The unitary evolution of the analog control Hamiltonian given is by $U(\bm{\theta})=\mathcal{T}\exp(-i\int_0^T H_c(\bm{\theta},\tau)d\tau)$, where $T$ is the total evolution time. 
Consequently, it is fully specified by the system and control Hamiltonians $H_0$ and $\{H_j\}$, the chosen basis functions $\{f_l\}$, the number of tunable parameters per control term $d$, and the total control strength $\theta_{\max} T$. 
The expressive power of the analog process is governed by both the maximum pulse amplitude $\theta_{\max}$ and the evolution time $T$, which together define the control strength $\theta_{\max}T$. The specific setup for UAQS can be found in Sec.~\ref{sec:ansatzsetting}.

The control Hamiltonian in Eq.~\eqref{Eq:analogH} provides a unified framework for describing analog Hamiltonians with varying degrees of controllability. In the simplest case, when only $H_0$ is present, the system reduces to a conventional analog simulator with a fixed Hamiltonian. The inclusion of control pulses $u_j$ captures the enhanced controllability available in modern analog quantum platforms. These controls can, for instance, represent arbitrary single-qubit operations realizable in superconducting qubit systems~\cite{kjaergaard2019superconducting,mezzacapo2014digital,garcia2015fermion,asaad2016independent,weber2017coherent}, trapped-ion platforms~\cite{zhang2017observation,haffner2005scalable,blatt2012quantum}, and Rydberg atom arrays~\cite{saffman2016quantum}. While these controls extend the class of simulable Hamiltonians within the form of Eq.~\eqref{Eq:analogH}, it remains a significant challenge to simulate more general Hamiltonians that cannot be expressed in this form.
Therefore, we present our UAQS framework to address this important limitation in a different route. Rather than compiling dynamics into gates~\cite{li2017efficient,cirstoiu2020variational,yao2021adaptive,heya2023subspace}, we embed optimization directly into native analog quantum control which continuously shapes time-dependent control fields to steer the quantum system toward desired dynamical trajectories without ever leaving the analog domain.

This principle is formalized in our framework of UAQS, which recasts both real-time and imaginary-time quantum evolution as continuous-time control optimization problems. The central conceptual shift is a reframing of the simulation problem: instead of considering how the target evolution can be decomposed into available gate operations, UAQS explorers how to continuously tune the available analog controls so that the parametrized analog simulator approximates the desired quantum dynamics as faithfully as possible. This reformulation naturally unifies the systematic, parameter-driven optimization paradigm of digital quantum algorithms with the continuous-time evolution that is the defining physical strength of analog platforms.

\subsection{Framework of UAQS}

To simulate general Hamiltonians, we introduce UAQS, a framework for programmable analog quantum hardware in which the quantum state is prepared via a parameterized analog Hamiltonian instead of a discrete sequence of quantum gates. 

The central idea is to restrict the quantum dynamics to a low-dimensional manifold $\mathcal{M}$ in Hilbert space, parameterized by tunable control parameters $\bm{\theta}(t)=(\theta_1(t),\ldots,\theta_K(t))$. The state $|\psi(\bm{\theta}(t))\rangle$, which approximates the true evolution, is generated by continuous evolution under a parameterized analog Hamiltonian $H_c(\bm{\theta},\tau)$, reducing the simulation to solving for the time dependence of the $K$ classical parameters.
The governing principle of the framework is that we project the exact quantum dynamics $d|\phi(t)\rangle/dt$, governed by the Schrödinger equation, onto the tangent space $\mathcal{T}_\psi \mathcal{M}$ of the control parameter manifold at the current state $|\psi(\bm{\theta}(t))\rangle$. This projection leads the optimal parameter update $\dot{\bm{\theta}}$ such that the  state $|\psi(\bm{\theta}(t))\rangle$ tracks the true dynamics as faithfully as possible within the expressive capacity of the control Hamiltonian. 
Concretely, the tangent space is spanned by the partial derivatives $\{|\partial_{\theta_k} \psi(\bm{\theta}(t))\rangle\}$, and the variational principle determines the parameter dynamics by minimizing the $L^2$-norm of the residual between the projected and exact time derivatives in Hilbert space. This yields a system of equations that govern the flow of parameters on $\mathcal{M}$ as shown in Fig.~\ref{fig:framework}, which we detail below for both real-time and imaginary-time evolution.\\
\paragraph{Real-time evolution} 
The real-time dynamics of a closed quantum system are governed by the time-dependent Schrödinger equation
\begin{equation}
\frac{d|\phi(t)\rangle}{dt} = -i H |\phi(t)\rangle,
\end{equation}
where $H$ is the target Hamiltonian. On analog quantum devices, it is not available to access arbitrary gate sequences, instead, the hardware evolves under a native Hamiltonian whose structure is fixed by the physical platform.
In our framework, the state $|\psi(\bm{\theta}(t))\rangle$ evolves under the parametrized analog Hamiltonian $H(\bm{\theta}(t))$, and the parameters $\bm{\theta}$ are updated according to the variational principle. Through enforcing the residual error of Schrödinger equation be orthogonal to the tangent space of the manifold $\mathcal{T}_\psi \mathcal{M}$, which is also equivalent to minimize the least-square distance between the time derivative of parameters and the exact Schrodinger dynamics, one obtains a evolution of parameters, i.e. $M^R\dot{\bm{\theta}}=V^I$, where $M^R \equiv \mathrm{Re}[M]$ and $V^I \equiv \mathrm{Im}[V]$, respectively.
\begin{align}
M_{i,j}& = \frac{\partial\left<\psi(\bm{\theta}(t))\right|}{\partial\bm{\theta}_i}\frac{\partial\left|\psi(\bm{\theta}(t))\right>}{\partial\bm{\theta}_j},\\
V_i &= \frac{\partial\left\langle\psi(\bm{\theta}(t))\right|}{\partial\bm{\theta}_{i}}H\left|\psi(\bm{\theta}(t))\right\rangle.
\end{align}
At each time step, the parameter update $\dot{\bm{\theta}}$ is determined by solving this linear equation $M^R \dot{\bm{\theta}} = V^I$. The updated parameters then define the next state, and such procedure is iterated to approximate the full real-time trajectory.
\paragraph{Imaginary-time evolution}
Imaginary-time evolution is a powerful technique for preparing ground states of quantum Hamiltonians. Formally, it is obtained by the replacement $t \to i\tau$ in the Schrödinger equation, yielding
\begin{equation}
\frac{d|\phi(\tau)\rangle}{dt} = -(H-E_\tau) |\phi(\tau)\rangle,
\end{equation}
where $E_\tau=\langle H\rangle_{|\phi(\tau)\rangle}$. Under this flow, any initial state with nonzero overlap with the ground state converges exponentially to the ground state as $\tau \to \infty$.
Applying the same projection of the imaginary-time evolved state, it also yields a linear equation for the evolution of parameters, $M^R \dot{\bm{\theta}} = -V^R$ where $V^R \equiv \mathrm{Re}[V]$. The structure of the equation is identical to the real-time case, with the same coefficient matrix $M$ on the left-hand side, but with the imaginary and real parts of $V$ exchanged and a sign reversal. This parallel structure means that both real and imaginary-time simulations can be implemented within the same framework, requiring the same underlying measurements of $M$ and $V$.\\

The practical bottleneck of quantum simulation over parameter space lies in the efficient estimate the coefficient matrix $M$ and $V$. On digital quantum computers, the matrix elements of $M$ and $V$ is not hard to estimate, which can be obtained by performing a Hadamard test like circuit that extracts real and imaginary parts of arbitrary inner products by introducing a single ancilla qubit and a controlled unitary operation. On analog quantum hardware, however, this approach is not directly applicable, because these discrete, controlled-unitary operations that fundamentally conflict with the continuous, global evolution of analog systems. Besides, implementing such operations on analog hardware would require significant additional overhead, including auxiliary physical qubits, engineered long-range couplings to serve as control channels, and precise timing of conditional interactions, all of which are technically challenging without introducing significant noise or decoherence. Thus, estimating $M$ and $V$ on analog hardware is therefore a nontrivial problem that requires a fundamentally different approach. To address this, we introduce estimation methods that require no ancilla qubits and no controlled unitary operations, relying solely on ingredients that are native to programmable analog hardware: continuous Hamiltonian evolution, intermediate projective measurements in Pauli bases, and local observable readout.

Given a target Hamiltonian $H$ and the controllable Hamiltonian $H_c(\bm{\theta},\tau)=\sum_{j=1}^m \bm{u}_j(\bm{\theta},\tau)H_j$, where $\bm{u}_j(\bm{\theta},\tau)$ is tunable shape function. For simulating real-time evolution, we approximate the matrix elements of $M^R$ and $V^I$ using Monte Carlo integration combined with a direct measurement protocol that requires no ancillary qubits and is compatible with analog platforms. The derivations are deferred to Appendix~\ref{sec:proof}.
\begin{align}\label{eq:M_V_R}
M_{j,k}^R & \approx -\frac{T^2}{N^2}\sum_{p=1}^N\sum_{q=1}^N \frac{\partial{\bm{u}_j}(\tau_p)}{\partial{\theta_j}}\frac{\partial{\bm{u}_k}(\tau_q)}{\partial{\theta_k}}\left(P_{j,k}^{+} - P_{j,k}^{-}\right)\nonumber \\
V_{j}^I & \approx -\frac{T}{N}\sum_{p=1}^N \frac{\partial{u_j}(\tau_p)}{\partial{\theta_j}} \left(P_{j}^{+} - P_{j}^{-}\right)
\end{align}
where $P_{j,k}^{\pm} = p(M_{H_j}=\pm1)\langle H_k\rangle_{\pm}$ and $P_{j}^{\pm} = p(M_{H_j}=\pm1)\langle H\rangle_{\pm}$. As $P_{j,k}^{\pm}$ and $P_j^{\pm}$ share the same structure and estimation procedure, we focus on the estimation of $P_{j,k}^{\pm}$ without loss of generality. In Eq.~\eqref{eq:M_V_R}, $\tau_p$, $\tau_q$ are sampled time used for Monte Carlo integration, and $P_{j,k}^{\pm} = p(M_{H_j}=\pm1)\langle H_k\rangle_{\pm}$ where $p(M_{H_j}=\pm 1)=\|\frac{1}{2}(I\pm H_j)|\psi(\tau_p)\rangle\|^2$ refers to the probability of getting the result $\pm 1$ via performing projective measurement $\mathcal{M}_{H_j}$ on $|\psi(\tau_p)\rangle$. $\langle H_k\rangle_{\pm}$ is the expectation value of $H_k$ over the state $|\psi(\tau_q)^{\pm}\rangle=U_{\bm{\theta}}(\tau_q,\tau_p)|\psi_{\pm}\rangle$ where $U_{\bm{\theta}}(\tau_q,\tau_p)$ is analog quantum evolution from $\tau_p$ to $\tau_q$ and $|\psi_{\pm}\rangle$ is the post projected state of $\mathcal{M}_{H_j}$.
Concretely, we begin by preparing the initial quantum state $|\psi_0\rangle$ and specifying the total evolution time $T$. To estimate the elements of $M^R$ via Monte Carlo integration, it is also required independently and uniformly draw time-stamp pairs $\mathcal{S}=\{\tau_p, \tau_q|\tau_p,\tau_q \in [0,T], p,q=1,\cdots,N\}$.
First, we evolve the initial state $|\psi_0\rangle$ from $0$ to $\tau_p$ under the programmable Hamiltonian $H_c$, thereby generating the state $|\psi(\tau_p)\rangle$. At this time, we need perform a nondestructive projective measurement of $H_j$ on $|\psi(\tau_p)\rangle$ which produces outcomes $M_{H_j}=\pm 1$ with corresponding probabilities $p(M_{H_j}=\pm 1)$ and post measurement state $|\psi_{\pm}\rangle$. Then, the projected state $|\psi_{\pm}\rangle$ are further evolved from $\tau_p$ to $\tau_q$ by the control Hamiltonian $H_c$, resulting in states $|\psi(\tau_q)^{\pm}\rangle=U_{\bm{\theta}}(\tau_q,\tau_p)|\psi_{\pm}\rangle$. Subsequently, the expectation values of the Hamiltonian term $H_k$ are evaluated on the post-measurement states $|\psi(\tau_q)^{\pm}\rangle$ yielding $\langle H_k\rangle_{\pm}$. At the end, these quantities combined with $p(M_{H_j}=\pm1)$ forms $P_{j,k}^{\pm}=p(M_{H_j}=\pm1)\langle H_k\rangle_{\pm}$.
Repeating the above procedure over the set $\mathcal{S}$ and averaging the weighted differences $P_{j,k}^{+}-P_{j,k}^{-}$, with weights $\partial_{\theta_j}\bm{u}_j(\tau_p)$ and $\partial_{\theta_k}\bm{u}_k(\tau_q)$ and prefactor $T^2/N^2$, yields an unbiased estimator of $M^{R}_{j,k}$.

For simulating imaginary-time evolution, we have the following form to estimate $V^I$, which shares the similar process of estimating $V^R$ but integrating phase shift rule with the Monte Carlo integration. 
\begin{equation}
V_{j}^R \approx -\frac{T}{2N}\sum_{p=1}^N \frac{\partial{\bm{u}_j}(\tau_p)}{\partial{\theta_j}} \cdot \left(Q_{j}^{+} - Q_{j}^{-}\right),  
\end{equation}
where $Q_{j}^{\pm} = \langle \psi^{\pm}|H|\psi^{\pm}\rangle$ and $|\psi^{\pm}\rangle = U_{\bm{\theta}} (T,\tau_p) e^{-i(\pm\frac{1}{4})\pi H_j}|\psi(\tau_p)\rangle$. 
Specifically, the procedure requires $N$ samples ${\tau_p}_{p=1}^N$. For each $\tau_p$, the initial state $|\psi_0\rangle$ is evolved to $|\psi(\tau_p)\rangle$ by $H_c$,  followed by a short-time evolution under $\pm H_j$ with $\Delta\tau=\pi/4$. The system is then evolved from $\tau_p$ to $T$, producing states $|\psi^{\pm}\rangle$ whose energy expectations yield $Q_j^{\pm}$. Averaging the weighted differences $Q_j^{+}-Q_j^{-}$ over all samples, with weights $\partial_{\theta_j}\bm{u}_j(\tau_p)$, provides an estimate of $V_j^{R}$.

Since the dynamics are expressed in terms of the parameter evolution $\bm{\theta}(t)$, the time evolution of system over a small interval $\delta t$, i.e. $|\phi_{t+\delta t}\rangle$, can be approximated using Euler’s method based on the estimated $\dot{\bm{\theta}}$.
\begin{equation}
|\phi_{t+\delta t}\rangle \approx \mathcal{T} \exp\bigg(-i\int_0^T H_c(\bm{\theta}_{t+\delta t},\tau)d\tau\bigg)|\psi_0\rangle,
\end{equation}
where $\bm{\theta}_{t+\delta t} = \bm{\theta}(t) + \dot{\bm{\theta}}\delta t$.

\section{Numerical Results}
In this section, we present universal analog quantum simulation results for both real and imaginary time evolution, together with two representative applications, including the estimation of transition energies and the simulating the out-of-time-ordered correlator.

\subsection{Setting of analog Hamiltonian}
\label{sec:ansatzsetting}
First, we introduce the setup of the parametrized analog Hamiltonian across different hardware platforms, where the distinctions are primarily reflected in the choice of the system Hamiltonian $H_0$. In the following experiments, we mainly focus on two platforms: Transmon and Rydberg atom arrays. For the Rydberg atom array, two system Hamiltonians can be realized by exploiting different choices of atomic energy levels \cite{georgescu2014quantum}. The corresponding constant system Hamiltonians for these systems are presented below,
\begin{itemize}
\item Rydberg atom array (Ising)\\
\begin{equation}
H_0 = \sum_{ij,i\neq j}\frac{C}{|i-j|^6}(I-Z_i)(I-Z_j).
\end{equation}
\item Rydberg atom array (XY)\\
\begin{equation}
H_0=2\sum_{ij,i\neq j}\frac{C}{|i-j|^3}\left(X_iX_j+Y_iY_j\right).
\end{equation}
\item 1D Transmon array\\
\begin{equation}
H_0=\sum_{i}\frac{C}{2}(X_i X_{i+1}+Y_iY_{i+1}).
\end{equation}
\end{itemize}
The controllable Hamiltonians are chosen as $H_j \in \{X, Z\}$. We set the interaction strength $C$ to be of the same order of magnitude as the maximum parameter value $\theta_{\max}=5C$. Additionally, the amplitude basis function $f_l$ is selected to be piecewise linear function defined over intervals of equal length.
\begin{equation}
f_l(t) = 
\left\{
\begin{aligned}
1, \qquad  & (l-1)\frac{T}{d} \leq t<l\frac{T}{d}  \\\
0. \qquad   & \textrm{others}
\end{aligned}
\right.
\end{equation}
where $d$ is the number of the basis function. With such basis function, the integral estimates of physical quantities associated with $\theta_{jl}$ need only be evaluated over the time interval $\big[(l-1)\tfrac{T}{d}, l\tfrac{T}{d}\big]$.

\subsection{Real time evolution}
First, we take the transverse-field Ising model (TFIM) as an example to validate the feasibility of performing UAQS under the different control Hamiltonians~\cite{fradkin2013field}. Here, we consider a 6-qubits TFIM as the target system,
\begin{equation}
H=J\sum_jZ_jZ_{j+1}-h\sum_iX_i.
\end{equation}
In the simulation, we fix the duration of the parametrized analog Hamiltonian to $T = 10\pi$, choose $20$ control parameters for manipulating each control Hamiltonian $H_j$. Consequently, the parametrized analog Hamiltonian involves approximately $120$ tunable parameters in total. Besides, we set $|+\rangle^{\otimes 6}$ as the initial state. The target system dynamics are simulated over the time interval $t \in [0,4]$ using a time step of $\delta t = 0.01$. To quantify the performance of UAQS, we evaluate the nearest-neighbor correlation between the simulated state and the exact evolved state,
\begin{equation}
C=\frac{1}{n-1}\sum_{i=1}^{n-1}Z_iZ_{i+1}
\end{equation}
We separately evaluate the proposed approach on two TFIM with $J=1,h=0.5$ and $J=1,h=1$, and the corresponding results are shown in Fig.~\ref{fig:ising}.

\begin{figure}[t]
    \centering
    \includegraphics[width=\linewidth]{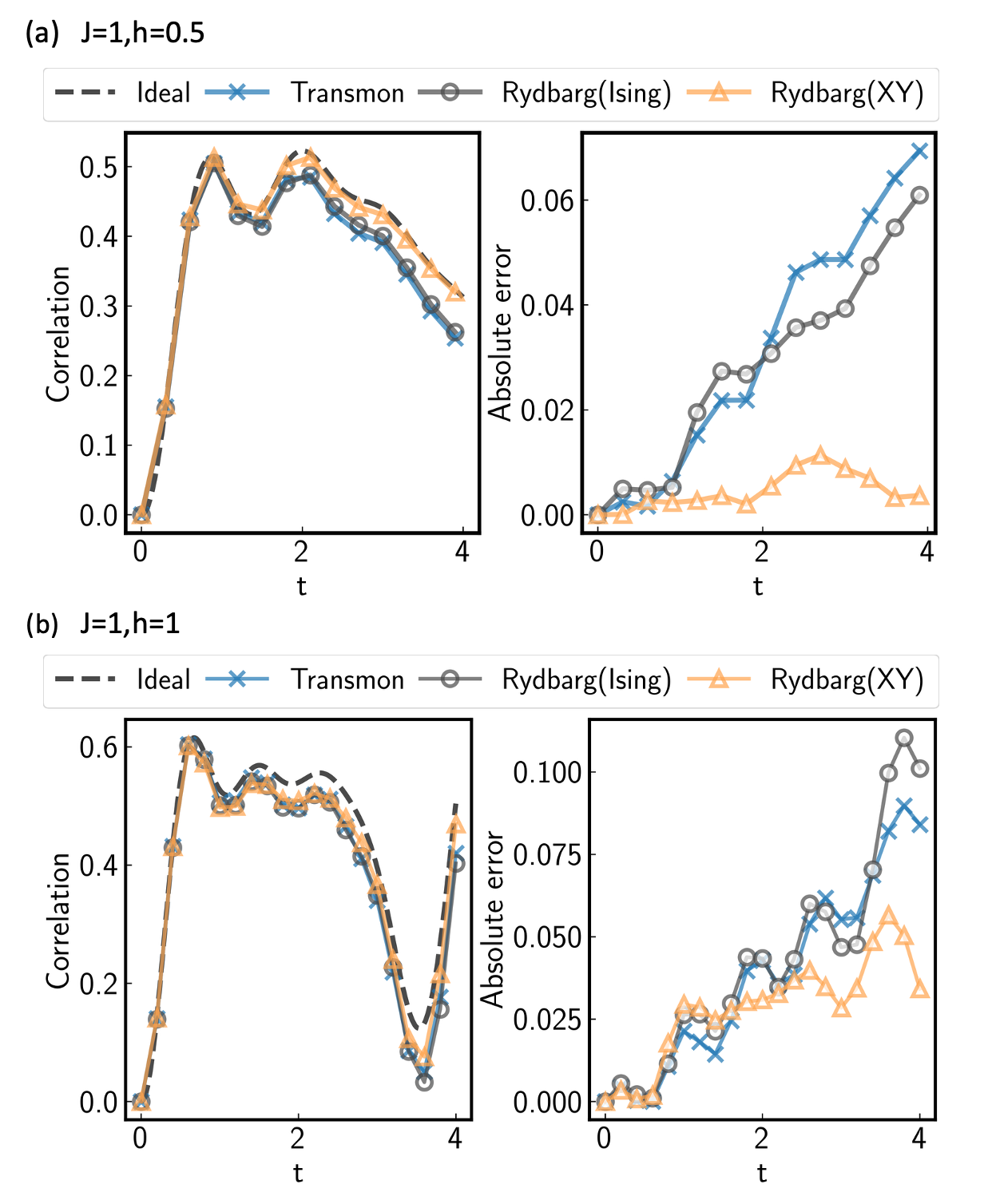}
    \caption{Simulation of the real-time dynamics of TFIM with (a) $J=1$ and $h=0.5$, (b) $J=1$ and $h=1$. The left panel displays the nearest-neighbor correlation of simulated and target system during the evolution, while the right panel reports the corresponding absolute error with respect to the exact dynamics. The dashed black curves represent the exact time evolution, whereas the colored markers correspond to UAQS implemented with different analog platforms, namely Transmon, Rydberg Ising configuration, and Rydberg XY configuration.}
    \label{fig:ising}
\end{figure}

For $h=0.5$, all simulations closely follow the exact correlation dynamics at early and intermediate times, accurately reproducing the correlation behavior. As the evolution time increases, errors gradually accumulate, leading to a nearly monotonic growth of the absolute error. Among the considered platforms, the Rydberg (XY) ansatz consistently yields the smallest error throughout the evolution, whereas the Transmon and Rydberg (Ising) ansatz exhibit larger but comparable deviations at later times. 

In the case of $h=1$, the dynamics become more pronounced, featuring sharper variations in the correlation function. Although the simulations remain accurate at short times, error compared with the exact evolution grows more rapidly than in the $h=0.5$ case, leading to larger absolute errors at late times. Nevertheless, all three analog Hamiltonian successfully capture the qualitative features of the dynamics. Similar to the weak-field case, the Rydberg (XY) ansatz achieves the highest overall accuracy, whereas the Transmon and Rydberg (Ising) ansatz exhibit more pronounced error growth. 

The outperformance of the Rydberg (XY) ansatz can be understood from the perspective of expressibility and controllability. Compared to Ising-type controls, the inclusion of both \textit{XX} and \textit{YY} interaction terms enlarges the dynamical Lie algebra generated by the available control Hamiltonians, which increases in Lie algebra rank expands the reachable manifold of unitaries within a fixed evolution time, thereby enhancing the expressibility of the analog Hamiltonian \cite{d2021introduction}. Besides, the presence of both $YY$ and 
$YY$ interaction terms effectively enlarges the accessible dynamical Lie algebra, improving the controllability of the system within a fixed evolution time. As a result, it can realize smoother and more faithful approximations to the target evolution, leading to reduced accumulated errors compared to Transmon or Rydberg Ising ansatz.
Overall, these results demonstrate that the UAQS can reliably approximate real-time TFIM dynamics across different platforms, with accuracy depending on both the target Hamiltonian and the expressive power of the underlying analog control.

\begin{figure}[t]
    \centering
    \includegraphics[width=0.9\linewidth]{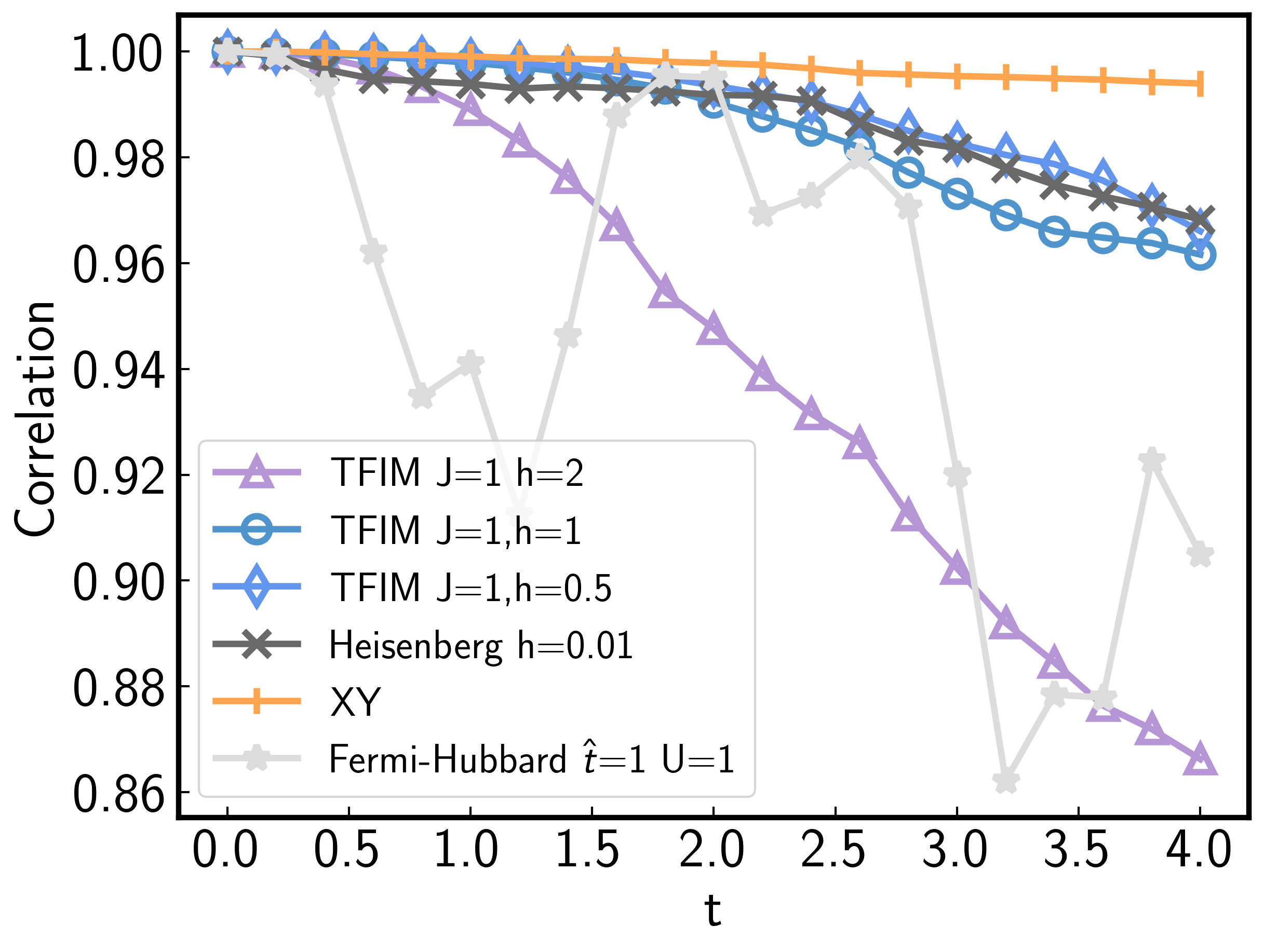}
    \caption{Numerical simulation of real-time evolution of different many-body quantum Hamiltonians.}
    \label{fig:fidelity_rxy}
\end{figure}

Then, we verify the performance of UAQS for simulating the dynamics of different quantum many-body systems based on above mentioned Rydberg (XY) ansatz. The target many-body systems includes 1) TFIM with $J=1,h=\{0.5,1,2\}$. 2) XY model, $H=\sum_{i}^n(X_iX_{i+1}+Y_iY_{i+1})/2$. 3) periodic Heisenberg model, $H=-\frac{1}{2}\sum_{i}^n(X_iX_{i+1}+Y_iY_{i+1}+Z_iZ_{i+1}+hZ_i)$ with $h=0.01$. 4) 1D Fermi-Hubbard model, $H=-\hat{t}\sum_{i,\sigma}\left(\hat{c}_{i,\sigma}^\dagger\hat{c}_{i+1,\sigma}+\hat{c}_{i+1,\sigma}^\dagger\hat{c}_{i,\sigma}\right)+U\sum_i\hat{n}_{i\uparrow}\hat{n}_{i\downarrow}$ with $\hat{t}=1, U=1$. We set the evolution time of target system from $t=0$ to $t=4$ with time interval $\delta t =0.01$. For parametrized analog Hamiltonian, we also adopt similar setup, that is setting $32$ parameters for each control Hamiltonian $H_j$ and total duration as $T=12\pi$. Again, we also choose $|+\rangle^{\otimes n}$ as the initial state. Here, we use fidelity between the simulated state and the exact evolved state as the metric to evaluate the performance of UAQS. The corresponding time-dependent behavior of fidelity are shown in Fig.~\ref{fig:fidelity_rxy}. 

For TFIM with $J=1$ and varying transverse field strengths $h$, the fidelity remains close to $1$ at short evolution times and exhibits a clear field-dependent decay as the dynamics become increasingly nontrivial. In particular, stronger transverse fields lead to a more rapid degradation of fidelity, indicating a growing mismatch between the target Hamiltonian and the restricted control manifold accessible to the analog Hamiltonian. We also notice that UAQS achieves consistently high fidelity when simulating models with compatible interaction structures, such as the XY model and the weak-field Heisenberg model, reflecting its enhanced expressibility in these cases. However, for the Fermi-Hubbard model, the fidelity shows pronounced temporal fluctuations and a faster decay,
which indicates that Rydberg XY ansatz only weakly overlaps with the full Lie algebra required to approximate fermionic hopping and interaction terms. Overall, these results demonstrate that the simulation accuracy of analog Hamiltonian is model-dependent and is governed by the structural compatibility between the target Hamiltonian and the dynamical Lie algebra generated by the available analog controls. 

\subsection{Imaginary time evolution}
In this section, we report UAQS simulation for imaginary-time evolution. As a benchmark, we consider the TFIM with $J=1,h=1$, and simulate its imaginary-time dynamics from $t=0$ to $t=1$ with a time step of $\delta t = 0.01$. The analog Hamiltonian is implemented with a total evolution duration $T = 10\pi$ and $20$ parameters for each control Hamiltonian $H_j$. The results are presented in Fig.~\ref{fig:imagine-time}, which shows the time-dependent energy function of imaginary time $\tilde{\tau}$. 

As shown in Fig.~\ref{fig:imagine-time}, the energy of UAQS decreases monotonically with increasing imaginary time $\tilde{\tau}$ and rapidly converges to the ground-state energy. Besides, the close agreement between the UAQS result and the exact curve indicates that the parametrized analog Hamiltonian accurately tracks the imaginary time trajectory and successfully reproduces the ground-state convergence behavior. The inset displays the fidelity between the our evolved state and the exact imaginary time evolved state. The fidelity remains close to $1$ over the entire evolution, while it has a shallow dip at intermediate times before towards $1$ later. This behavior indicates that although transient the simulation errors arise during the early and intermediate stages, they do not accumulate and are effectively suppressed as the state relaxes toward the ground-state manifold.
In summary, the results demonstrate that UAQS are capable of approximating the imaginary-time dynamics of the TFIM and maintaining high fidelity over the entire evolution.

\begin{figure}[t]
    \centering
    \includegraphics[width=0.9\linewidth]{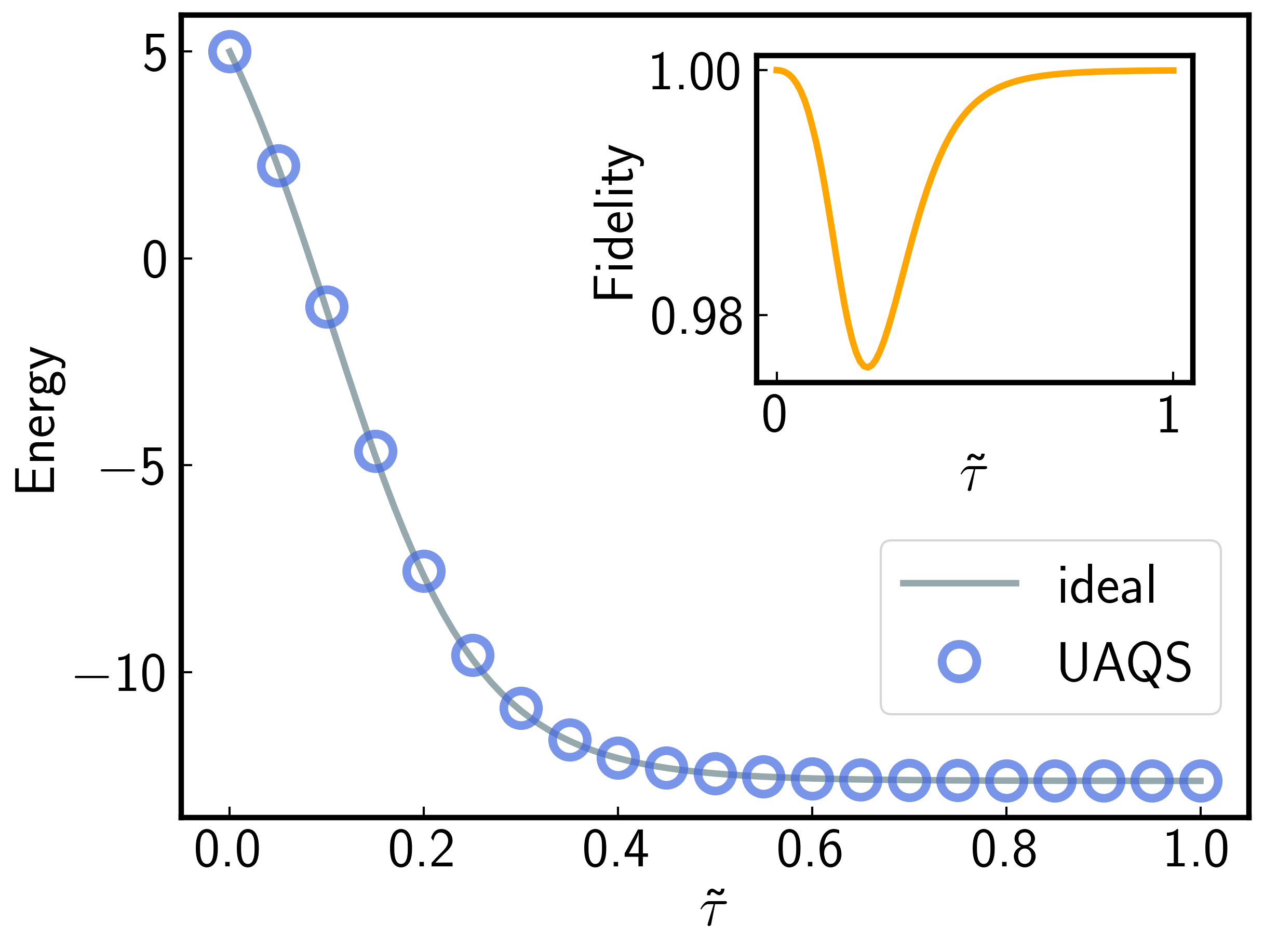}
    
    \caption{Performance of UAQS for imaginary-time evolution of the TFIM $J=1, h=1$. The solid gray line shows the exact imaginary-time evolution, while blue circular markers denote the UAQS results. The inset reports the fidelity between the simulated state and the exact state during the evolution.
}
    \label{fig:imagine-time}
\end{figure}

\subsection{Applications}
In this section, we explore two representative applications of UAQS including transition energy estimation and the simulation of out-of-time-order correlators.

\subsubsection{Estimating transition energy}
As the transition energies between eigenstates of a quantum many-body system can be determined via dynamical simulations~\cite{sun2025probing,yang2024resource,huang2026low}, we explore the UAQS to estimate the transition energy. Specifically, we focus on the following physical quantity,
\begin{equation}
G(\omega)=\sum_{n,n'}\langle n|\rho|n'\rangle \langle n'|\hat{O}|n\rangle e^{-\tau^2 (E_n-E_{n'}-\omega)^2}
\end{equation}
where $|n\rangle $ and $|n'\rangle$ represent the different eigenstates of a given many-body Hamiltonian $H$, and $E_n$ and $E_{n'}$ are the corresponding eigenvalues, respectively.
$G(\omega)$ is a spectral function, that defines a frequency-resolved response associated with an operator $\hat O$ and resolves transition energies $\Delta_{n,n'}=E_n-E_{n'}$, evaluated with respect to a quantum state $\rho$.  When $\tau$ is sufficiently large and the observable $\hat{O}$ is appropriately chosen, $G(\omega)$ takes its local maximum when $\omega$ approaches the transition energies $\Delta_{n,n'}$, which can be viewed as a spectral detector.

Thus, we attempt to calculate the $G(\omega)$ by UAQS of a real time dynamic. Assuming UAQS with a step size of $\Delta t$, then we have the following approximation,
\begin{equation}
G(\omega)\approx\frac{1}{\sqrt{\pi}}\sum_{n=0}^N e^{-\frac{(n\Delta t)^2}{4}}\cos(\tau\omega n\Delta t) \mathrm{Tr}[\rho(\tau n\Delta t)\hat{O}].
\end{equation}

To evaluate the performance of UAQS, we consider test it on 5-qubits Heisenberg models with periodic boundary condition under different external field,
\begin{equation}
H = J\sum_i (X_iX_{i+1}+Y_iY_{i+1}+Z_iZ_{i+1})+h_z\sum_iZ_i
\end{equation}
Specifically, for both cases, the dynamical evolution is simulated using UAQS from $t=0$ to $t=16$ on initial state $|+\rangle^{\otimes 5}$ with a time step $\delta t = 0.01$. We adopt the Transom system to build the parameterized analog Hamiltonian with the duration $T = 10\pi$ and 20 control parameters. 

\begin{figure}[t]
    \centering
    \includegraphics[width=\linewidth]{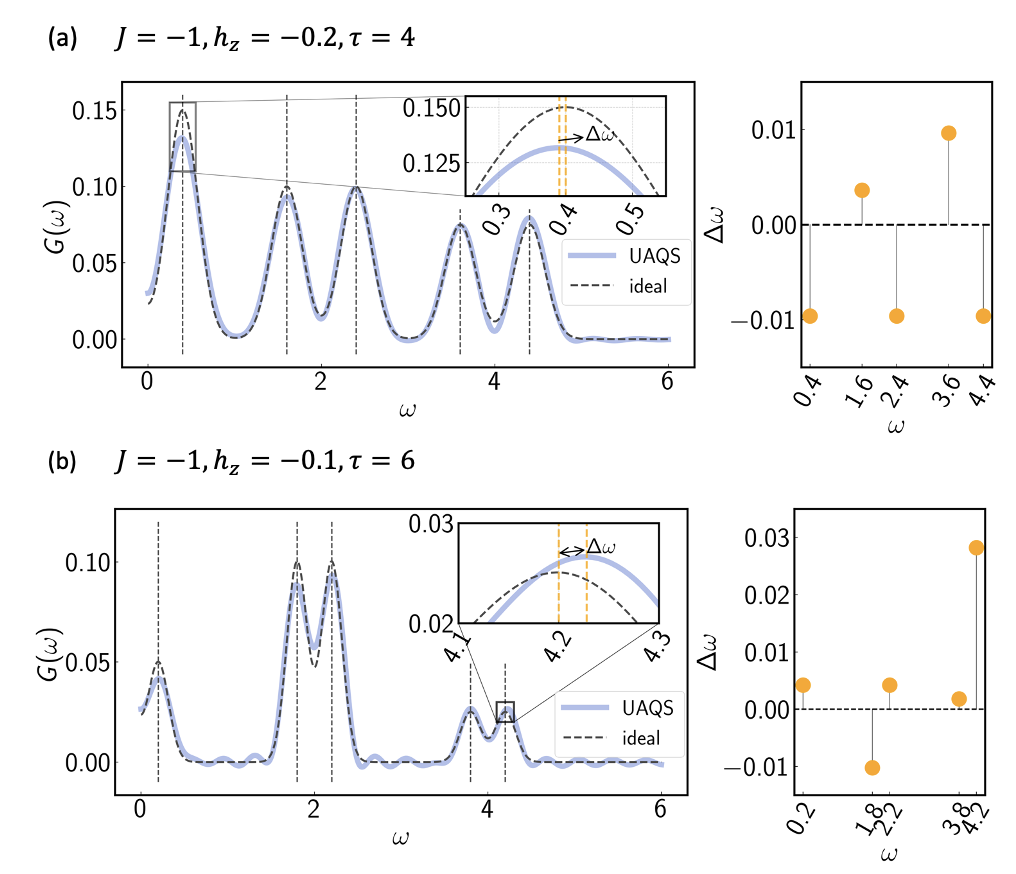}
    \caption{The estimated transition energy spectrum of the Heisenberg model. The subfigure (a) and (b) present the result of the system with $J=-1,h=0.2,\tau=4$ and $J=-1,h=0.1,\tau=6$, separately. In each left panel, the solid light-blue curve represents the spectrum reconstructed from UAQS simulation, while and the vertical dashed lines indicate the exact transition energies and the black dashed curve corresponds to the ideal spectrum. The right panel shows the the error $\Delta \omega$ between the UAQS estimated transition energies and the exact values for each identified transition.
}
    \label{fig:transition_energy}
\end{figure}

Moreover, we take $\hat{O} = 1/5\sum_{j=1}^5 X_j$, $\tau=4$ and $\tau=6$ for case of $h=0.2$ and $h=0.1$, separately. The transition energy spectrum obtained from UAQS is compared with the ideal result, as shown in Fig.~\ref{fig:transition_energy}. 
For both cases, we observed that UAQS accurately reproduces the overall spectral structure, including the number, locations, and relative dominant peaks, which shows a good agreement between the UAQS and ideal dynamics. The error is on the order of $10^{-2}$, with both positive and negative shifts, which indicates UAQS reliably captures the spectral features such that it provides reasonable accuracy in estimating the transition energies.

\subsubsection{Simulating out-of-time order correlator}
An out-of-time-order correlator (OTOC) is a diagnostic quantity used to quantify information scrambling~\cite{fan2017out,swingle2018unscrambling,garttner2017measuring}, operator growth, and quantum chaos in many-body quantum systems. Given two local operators $\hat{W}$ and $\hat{V}$, the OTOC of a pure state $|\psi_0\rangle$ is defined as,
\begin{equation}
    C(t)=\langle\psi_0|\hat{W}^\dagger(t)\hat{V}^{\dagger}\hat{W}(t)\hat{V}|\psi_0\rangle
\end{equation}
where $H$ is the system Hamiltonian and $\hat{W}(t)=e^{i\hat{H}t}\hat{W}e^{-i\hat{H}t}$. 
In general, OTOCs are challenging to measure. Nevertheless, when $\hat{W}$ and $\hat{V}$ are selected as single-qubit Pauli operators and the initial state $|\psi_0\rangle$ is prepared as an eigenstate of $\hat{V}$, the OTOC becomes accessible via straightforward measurements following the dynamical evolution.
Assuming we are capable of efficiently preparing the eigenstate $|\psi_0\rangle$ of $\hat{V}$, such as $V=H^{\otimes n}, |\psi_0\rangle=|+\rangle^{\otimes n}$, thus we have,
\begin{align}
C(t)&=\langle\psi_0|e^{-i\hat{H}t}\hat{W}^\dagger e^{i\hat{H}t}\hat{V}e^{i\hat{H}t}\hat{W}e^{-i\hat{H}t}|\psi_0\rangle \nonumber \\
&=\langle\psi(t)|\hat{V}|\psi(t)\rangle
\end{align}
where $|\psi(t)\rangle=e^{i\hat{H}t}\hat{W}e^{-i\hat{H}t}|\psi_0\rangle$. 
Here, we employ UAQS to implement the state $ |\psi(t)\rangle$ by using a parametrized analog Hamiltonian to evolve the initial state $|\psi_0\rangle$ under the Hamiltonian $H$, subsequently applying the local operator $\hat{W}$, and finally evolving the system under the reversed Hamiltonian $-H$. Specifically, we select a 6-qubits Heisenberg model as an example, i.e. $H=J\sum_{i=1}^n(X_iX_{i+1}+Y_iY_{i+1})/2$, where $J=1$. Additionally, to easily measure OTOCs and shows the phonomenion of probing operator spreading, we separately consider the OTOCs with $\hat{W}=\hat{\sigma}_n^x, \{\hat{V}=\hat{\sigma}_j^x|j=1,\dots,n\}$ and $\hat{W}=\hat{\sigma}_n^z, \{\hat{V}=\hat{\sigma}_j^z|j=1,\dots,n\}$ where $\hat{W}$ is a single operator acting on the last qubit, and $\hat{V}$ are operators acting on each qubit. Thus, we donate the OTOC with $\hat{W}=\hat{\sigma}^x_n$ and $\hat{V}=\hat{\sigma}_j^x$ as $C^{XX}_j(t)$, similarly, the other case is $C^{ZZ}_j(t)$. Here, we simulate the dynamical evolution of the Heisenberg model on $|\psi_0\rangle=|+\rangle^n$ and $|\psi_0\rangle=|010101\rangle$ with $4$ total evolution time and a step size of 0.01 by UAQS for $C_j^{XX}$ and $C_j^{ZZ}$, respectively. The corresponding results are presented in Fig.~\ref{fig:otoc}. 

\begin{figure}[t]
\centering
    \centering
    \includegraphics[width=\linewidth]{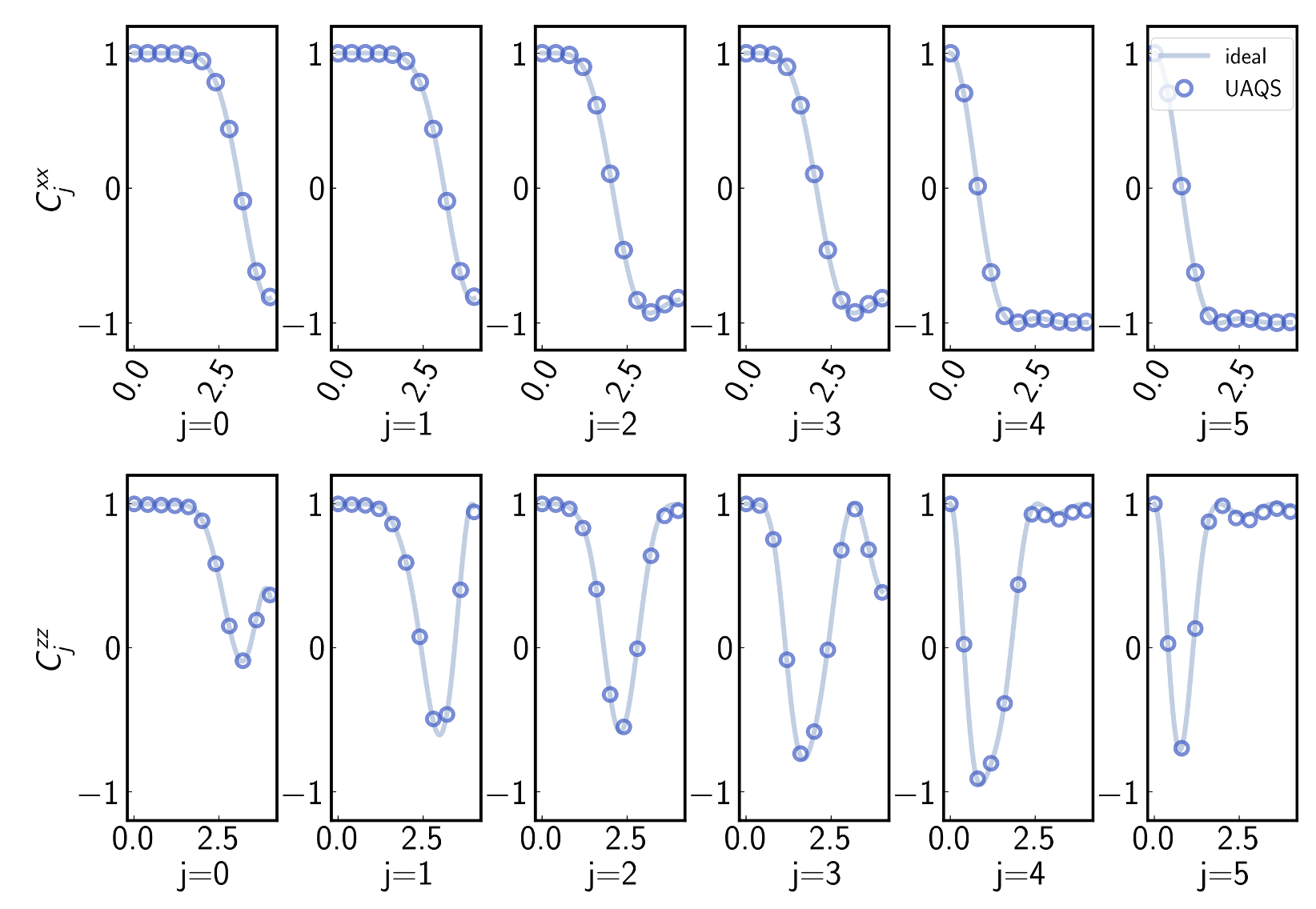}
    \caption{The experimental results of the site-resolved dynamics of OTOC for two different local operators, i.e. $C_j^{xx},C_j^{zz}$. The solid curves in each subfigure denote the exact results, and the open circles represent the UAQS results. The top row shows the correlator $C_j^{xx}(t)$ while the bottom row shows $C_j^{zz}(t)$. Each column corresponds to a lattice site $j$ ranging from $j=0$ to $j=n$.}
    \label{fig:otoc}
\end{figure}

The Fig.~\ref{fig:otoc} demonstrates the spatiotemporal evolution of OTOCs for two different local operators that reveals clear signatures of information scrambling and operator growth. In both cases, the correlators are initially close to $1$, which indicates the operators commute at early times, and subsequently decay as time evolves. The onset of decay occurs earlier for sites closer to $n$ and is progressively delayed for more distant sites that demonstrates a light-cone-like propagation of information with a finite spreading velocity and directly reflects the causal structure of many-body dynamics. In all panels, the close agreement between the UAQS results and the ideal evolution curves confirms that such simulation accurately captures both the spatially resolved scrambling dynamics and the operator-dependent features of the OTOCs.

\section{Conclusion}
In this work, we introduce a systematic framework for expanding the dynamical capabilities of fixed analog quantum platforms. By embedding optimization directly into continuous-time quantum control, UAQS bridges the conceptual gap between digital variational algorithms and native analog evolution. Rather than relying on gate decomposition or universal logical operations, our approach formulates both real-time dynamics and state preparation as control optimization problems executed within hardware-defined Hamiltonians. This parametrized pulse paradigm preserves the efficiency and low-depth advantages of analog simulation while significantly enhancing its effective expressibility. The accompanying gradient-evaluation strategy further enables stable and scalable parameter optimization in experimentally realistic settings. Together, these results establish UAQS as a versatile and hardware-compatible route toward programmable analog quantum simulation, providing a practical pathway for exploring complex many-body dynamics and non-equilibrium phenomena on near-term quantum devices.

\begin{acknowledgments}
This work is supported by Quantum Science and Technology-National Science and Technology Major Project (2023ZD0300200), the National Natural Science Foundation of China (Grant Nos.~U2330201, 22303005), and Beijing Natural Science Foundation Z250004, 
Beijing Science and Technology Planning Project (Grant No.~Z25110100810000),. The numerical experiments of this work are supported by the High-performance Computing Platform of Peking University.
\end{acknowledgments}

\bibliographystyle{unsrtnat}
\bibliography{mybib.bib}

\clearpage
\newpage
\onecolumngrid
\appendix

\section{Proof of theorems}\label{sec:proof}
In  this section, we provide the theoretical foundations and technical details underlying the main results presented in the text. We first review the variational principle that forms the basis of the proposed framework, clarifying its formulation and the assumptions required for its application to continuous-time quantum dynamics. Building on this principle, we then present a detailed proof of the theorems, with particular emphasis on the explicit estimation of the quantities $M$ and $V$ introduced in the main text.

\subsection{Variational principle}
Consider the real time dynamics of a many-body system which follows the Hamiltonian $H$, the time evolution is governed by the Schrödinger equation,
\begin{equation}
\frac{d|\phi\rangle}{dt}=-iH|\phi\rangle,
\end{equation}
In the variational framework, instead of directly calculate state $|\phi(t)\rangle$, we approximate it with parametrized trial state $|\psi(\bm{\theta}_t)\rangle$ which is generated by evolving the $|\psi_0\rangle$ with control Hamiltonian $H(\bm{\theta}_t)$ from $0$ to $T$.

\paragraph{\bf{Real time evolution}}
First, we introduce how to apply the McLachlan’s variational principle to UAQS. The main idea is to project the original evolution to the parametrized evolution and then minimize the distance over the derivative of the parameters,
\begin{equation}\label{eq:mprinciple_real}
\delta\left\|\left( \frac{d}{dt}+iH \right)|\psi(\bm{\theta}_t)\rangle\right\| = 0,
\end{equation} 
where $\|\rho\|$ is the trace norm of a state $\rho$, $\delta(x)$ is the delta function. Since the square of $\|\rho\|$ does not effect the variation, we now consider the $\|\rho\|^2$ for simplifying the calculation.
Let $L = \left\|\left( \frac{d}{dt}+iH \right)|\psi(\bm{\theta}_t)\rangle\right\|^2$ and $\dot{\theta_j}$ be the derivative of $\theta_j$ with respect to $t$, 
\begin{align}
L &=\left(\left(d/dt+iH\right)|\psi(\bm{\theta}_t\rangle)\right)^{\dagger}(d/dt+iH)\left|\psi(\bm{\theta}_t)\right\rangle,   \\
&=\sum_{i,j}\frac{\partial\left<\psi(\bm{\theta}_t)\right|}{\partial\theta_i}\frac{\partial\left|\psi(\bm{\theta}_t)\right>}{\partial\theta_j}\dot{\theta}_i^*\dot{\theta}_j+i\sum_i\frac{\partial\left<\psi(\bm{\theta}_t)\right|}{\partial\theta_i}H\left|\psi(\bm{\theta}_t)\right>\dot{\theta}_i^* \\
&-i\sum_i\left<\psi(\bm{\theta}_t)\right|H\frac{\partial\left|\psi(\bm{\theta}_t)\right>}{\partial\theta_i}\dot{\theta}_i+\left<\psi(\bm{\theta}_t)\right|H^2\left|\psi(\bm{\theta}_t)\right>.
\end{align}
\begin{itemize}
\item Considering $\bm{\theta} \in \mathbb{C}$, we have
\begin{align}
\frac{\partial{L}}{\partial\dot{\theta}_i}& =\left(\sum_j\frac{\partial\left<\psi(\bm{\theta}_t)\right|}{\partial\theta_i}\frac{\partial\left|\psi(\bm{\theta}_t)\right>}{\partial\theta_j}\dot{\theta}_j+i\frac{\partial\left<\psi(\bm{\theta}_t)\right|}{\partial\theta_i}H\left|\psi(\bm{\theta}_t)\right>\right)\delta\dot{\theta}_i^*,  \\
&+\left(\sum_j\frac{\partial\left<\psi(\bm{\theta}_t)\right|}{\partial\theta_j}\frac{\partial\left|\psi(\bm{\theta}_t)\right>}{\partial\theta_i}\dot{\theta}_j^*-i\frac{\partial\left<\psi(\bm{\theta}_t)\right|}{\partial\theta_i}H\left|\psi(\bm{\theta}_t)\right>\right)\delta\dot{\theta}_i.
\end{align}
According to Eq.\eqref{eq:mprinciple_real}, we have the 
\begin{equation}
\sum_j\frac{\partial\left<\psi(\bm{\theta}_t)\right|}{\partial\theta_i}\frac{\partial\left|\psi(\bm{\theta}_t)\right>}{\partial\theta_j}\dot{\theta}_j=-i\frac{\partial\left<\psi(\bm{\theta}_t)\right|}{\partial\theta_i}H\left|\psi(\bm{\theta}_t)\right>.
\end{equation}
Equivalently, it can be expressed as linear equation of parameters,
\begin{equation}
\sum_jM_{i,j}\dot{\theta}_j=-iV_i,
\end{equation}
where $M_{i,j}$ and $V_i$ is defined as,
\begin{align}
\label{eq:Mr}
M_{i,j}& = \frac{\partial\left<\psi(\bm{\theta}_t)\right|}{\partial\theta_i}\frac{\partial\left|\psi(\bm{\theta}_t)\right>}{\partial\theta_j}\\
V_i &= \frac{\partial\left\langle\psi(\boldsymbol{\theta}_t)\right|}{\partial\theta_{i}}H\left|\psi(\boldsymbol{\theta}_t)\right\rangle 
\end{align}
\item Considering $\bm{\theta}\in \mathbb{R}$, we have
\begin{align}
\label{eq:Vi}
\frac{\partial{L}}{\partial\dot{\theta}_i}& = \sum_j\left(\frac{\partial\left<\psi(\bm{\theta}_t)\right|}{\partial\theta_i}\frac{\partial\left|\psi(\bm{\theta}_t)\right>}{\partial\theta_j}+\frac{\partial\left<\psi(\bm{\theta}_t)\right|}{\partial\theta_j}\frac{\partial\left|\psi(\bm{\theta}_t)\right>}{\partial\theta_i}\right)\dot{\theta}_j\delta\dot{\theta_i}\\
&-i\left(\frac{\partial\left<\psi(\bm{\theta}_t)\right|}{\partial\theta_i}H\left|\psi(\bm{\theta}_t)\right>-\left<\psi(\bm{\theta}_t)\right|H\frac{\partial\left|\psi(\bm{\theta}_t)\right>}{\partial\theta_i}\right)\delta\dot{\theta_i}.
\end{align}
Similarly, according to Eq.\eqref{eq:mprinciple_real} we have the following linear equation of parameters,
\begin{equation}
\sum_jM_{i,j}^R\dot{\theta}_j=V_i^I,
\end{equation}
where $M_{i,j}^R$ and $V_i^I$ are the real and imaginary parts of $M_{i,j}$ and $V_i$, respectively.
\end{itemize}

\paragraph{\bf{Imaginary time evolution}}
For the imaginary time evolution, by applying McLachlan's variational principle, we have
\begin{equation}\label{eq:mprinciple_img}
\delta\left\|\left( \frac{d}{d\tau}+H - E_\tau \right)|\psi(\bm{\theta}_\tau)\rangle\right\| = 0,
\end{equation}
Similarly, let $L = \left\|\left( \frac{d}{d\tau}+H - E_\tau \right)|\psi(\bm{\theta}_\tau)\rangle\right\|^2$, then we have
\begin{align}
L &= \left(\left(d/d\tau+H - E_\tau\right)|\psi(\bm{\theta}_\tau\rangle)\right)^{\dagger}\left(d/d\tau+H - E_\tau\right)|\psi(\bm{\theta}_\tau\rangle)\\
&=\sum_{i,j}\frac{\partial\left<\psi(\bm{\theta}_\tau)\right|}{\partial\theta_i}\frac{\partial\left|\psi(\bm{\theta}_\tau)\right>}{\partial\theta_j}\dot{\theta}_i^*\dot{\theta}_j+\sum_i\frac{\partial\left<\psi(\bm{\theta}_\tau)\right|}{\partial\theta_i}(H - E_\tau)\left|\psi(\bm{\theta}_\tau)\right>\dot{\theta}_i^* \\
&+\sum_i\left<\psi(\bm{\theta}_\tau)\right|(H - E_\tau)\frac{\partial\left|\psi(\bm{\theta}_\tau)\right>}{\partial\theta_i}\dot{\theta}_i+\left<\psi(\bm{\theta}_\tau)\right|(H - E_\tau)^2\left|\psi(\bm{\theta}_\tau)\right>.
\end{align}
\begin{itemize}
\item Considering $\bm{\theta}\in \mathbb{C}$, we have
\begin{align}
\frac{\partial{L}}{\partial\dot{\theta}_i}& =\left(\sum_j\frac{\partial\left<\psi(\bm{\theta}_\tau)\right|}{\partial\theta_i}\frac{\partial\left|\psi(\bm{\theta}_\tau)\right>}{\partial\theta_j}\dot{\theta}_j+\frac{\partial\left<\psi(\bm{\theta}_\tau)\right|}{\partial\theta_i}(H - E_\tau)\left|\psi(\bm{\theta}_\tau)\right>\right)\delta\dot{\theta}_i^*,  \\
&+\left(\sum_j\frac{\partial\left<\psi(\bm{\theta}_\tau)\right|}{\partial\theta_j}\frac{\partial\left|\psi(\bm{\theta}_\tau)\right>}{\partial\theta_i}\dot{\theta}_j^*+\frac{\partial\left<\psi(\bm{\theta}_\tau)\right|}{\partial\theta_i}(H - E_\tau)\left|\psi(\bm{\theta}_\tau)\right>\right)\delta\dot{\theta}_i.
\end{align}
According to Eq.\eqref{eq:mprinciple_img}, we have the 
\begin{equation}
\sum_j\frac{\partial\left<\psi(\bm{\theta}_\tau)\right|}{\partial\theta_i}\frac{\partial\left|\psi(\bm{\theta}_\tau)\right>}{\partial\theta_j}\dot{\theta}_j=-\frac{\partial\left<\psi(\bm{\theta}_\tau)\right|}{\partial\theta_i}(H - E_\tau)\left|\psi(\bm{\theta}_\tau)\right>.
\end{equation}
Equivalently, it can be written as linear equation of parameters,
\begin{equation}
\sum_jM_{i,j}\dot{\theta}_j = V_i,
\end{equation}
where $M_{i,j}$ and $V_i$ is defined as,
\begin{align}
\label{eq:Mr}
M_{i,j}& = \frac{\partial\left<\psi(\bm{\theta}_\tau)\right|}{\partial\theta_i}\frac{\partial\left|\psi(\boldsymbol{\theta}_\tau)\right>}{\partial\theta_j},\\
V_i &= \frac{\partial\left\langle\psi(\bm{\theta}_\tau)\right|}{\partial\theta_{i}}(H - E_\tau)\left|\psi(\bm{\theta}_\tau)\right\rangle.
\end{align}
\item Considering $\bm{\theta}\in \mathbb{R}$, we have
\begin{align}
\frac{\partial{L}}{\partial\dot{\theta}_i}& =\sum_j\left(\frac{\partial\left<\psi(\bm{\theta}_\tau)\right|}{\partial\theta_i}\frac{\partial\left|\psi(\bm{\theta}_\tau)\right>}{\partial\theta_j}+\frac{\partial\left<\psi(\bm{\theta}_\tau)\right|}{\partial\theta_j}\frac{\partial\left|\psi(\bm{\theta}_\tau)\right>}{\partial\theta_i}\right)\dot{\theta}_j\delta\dot{\theta}_i,  \\
&+\left(\frac{\partial\left<\psi(\bm{\theta}_\tau)\right|}{\partial\theta_i}(H - E_\tau)\left|\psi(\bm{\theta}_\tau)\right>+\left<\psi(\bm{\theta}_\tau)\right|(H - E_\tau)\frac{\partial\left|\psi(\bm{\theta}_\tau)\right>}{\partial\theta_i}\right)\delta\dot{\theta}_i.
\end{align}

According to Eq.\eqref{eq:mprinciple_img}, we have the
\begin{align}
&\sum_j\left(\frac{\partial\left<\psi(\bm{\theta}_\tau)\right|}{\partial\theta_i}\frac{\partial\left|\psi(\bm{\theta}_\tau)\right>}{\partial\theta_j}+\frac{\partial\left<\psi(\bm{\theta}_\tau)\right|}{\partial\theta_j}\frac{\partial\left|\psi(\bm{\theta}_\tau)\right>}{\partial\theta_i}\right)\dot{\theta}_j\\
&=-\left(\frac{\partial\left<\psi(\bm{\theta}_\tau)\right|}{\partial\theta_i}(H - E_\tau)\left|\psi(\bm{\theta}_\tau)\right>+\left<\psi(\bm{\theta}_\tau)\right|(H - E_\tau)\frac{\partial\left|\psi(\bm{\theta}_\tau)\right>}{\partial\theta_i}\right).
\end{align}
Then, it corresponds to a linear equation of parameters,
\begin{equation}
\sum_jM_{i,j}^R\dot{\theta}_j = -V_i^R,
\end{equation}
Besides, as the parametrized states $\left|\psi(\bm{\theta}_\tau)\right>$ is normalized, thus we have $\Re[\frac{\partial\left<\psi(\bm{\theta}_\tau)\right|}{\partial\theta_i}E_\tau\left|\psi(\bm{\theta}_\tau)\right>]=0$.
\end{itemize}

\subsection{Estimate M and V}
In this section, we first provide the related lemmas and theorems to derivative the main theorems for estimating the $M$ and $V$, and then give the corresponding pseudo-codes. Since the parameters $\bm{\theta}$ in tunable pulses are real in general, we then only focus on the methods over the real parameters space.

\begin{theorem}[\bf{deravative of time evolution operator}~\cite{leng2022differentiable,kottmann2023evaluating}]\label{thm:deravative_evo_opt}
Given the Hamiltonian $H(\bm{\theta},t)= \sum_{j=1}^m \bm{u}_j(\bm{\theta},t)H_j$, where $\{H_j\}$ are the controllable Hamiltonian. The derivative of time evolution operator $U_{\bm{\theta}}(t_2,t_1)$ for time interval $[t_1, t_2]$ under Hamiltonian $H(\bm{\theta},t)$ with respect to parameters $\bm{\theta}_t$ is
\begin{equation}
\frac{\partial{U_{\bm{\theta}}}(t_2,t_1)}{\partial{\theta_j}} = -i\int_{t_1}^{t_2}U_{\bm{\theta}}(t_2,\tau)H_j U_{\bm{\theta}}(\tau,t_1)\cdot \frac{\partial{\bm{u}_j}}{\partial{\theta_j}}d\tau.
\end{equation}
\end{theorem}


\begin{theorem}[\bf{direct measurement protocol}~\cite{mitarai2019methodology}]\label{thm:measure}
The following protocol can be used to estimate the real part of the quantity $\langle\psi|W^\dagger UWG|\psi\rangle$, where $W$ is unitary operator, $U$ and $G$ are the Pauli string, i.e. $U,G\in \{I,\sigma_x, \sigma_y, \sigma_z\}^{\otimes n}$, 
\begin{equation}
\Re[\langle\psi|W^\dagger UWG|\psi\rangle] = p(M_G=+1)\langle U\rangle_{M_G=+1}-p(M_G=-1)\langle U\rangle_{M_G=-1}
\end{equation}
where $p(M_G=\pm 1)$ is the probability of getting $M_G=\pm 1$ by performing projective measurement $\mathcal{M}_G$ of $G$ on $|\psi\rangle$, where $\mathcal{M}_G$ can be performed by first transferring G to a single qubit $\sigma_z$ and then performing a measurement on the qubit nondestructively. $\langle U\rangle_{M_G=\pm 1}$ are the expectation value of $U$,
\begin{equation}
\langle U\rangle_{M_G=\pm1}=\frac1{4p(M_G=\pm1)}\left\langle\psi|(I\pm G)W^\dagger UW(I\pm G)|\psi\right\rangle 
\end{equation}
\end{theorem}

\begin{lemma}[\bf{generalized parameter phase shift rule}~\cite{wierichs2022general}]
\label{lemma:phase_shift_rule}
Given an observable $\hat{O}$ and the Hermitian generator $H$ of parametrized unitary $U(\theta)=e^{-i\frac{\theta}{2}H}$ has eigenvalues $\{\pm \lambda, 0\}$, then we can estimate the derivative of function $f(\theta)=\langle \psi|U^\dagger_{\theta}\hat{O}U_{\theta}|\psi\rangle$ with respect to $\theta$ by
\begin{align}
\partial_{\theta}f = \frac{\lambda}{2}\left[2\left(f(\theta + \frac{\pi}{4\lambda}) - f(\theta - \frac{\pi}{4\lambda} )\right) +(1-\sqrt{2})\left(f(\theta+\frac{\pi}{2\lambda}) - f(\theta-\frac{\pi}{2\lambda})\right)\right]
\end{align}
\end{lemma}

The theorem \ref{thm:measure} and lemma \ref{lemma:phase_shift_rule} are used for guiding us to directly estimate the terms in $M$ and $V$.

\begin{theorem}[{\bf{estimation of $M_{j,k}^R$}}]\label{thm:M_R}
 Given the controllable Hamiltonian $H(\bm{\theta},\tau)=\sum_{j=1}^m\bm{u}_j(\bm{\theta},\tau)H_j$, where $\{H_j\}_{j=1}^m$ are Pauli strings. The real part of $M_{i,j}=\frac{\partial\left<\psi(\bm{\theta})\right|}{\partial\theta_j}\frac{\partial\left|\psi(\bm{\theta})\right>}{\partial\theta_k}$ can be estimate by
\begin{equation}
  M_{j,k}^R  \approx -\frac{T^2}{N^2}\sum_{p=1}^N\sum_{q=1}^N \frac{\partial{\bm{u}_j}(\tau_p)}{\partial{\theta_j}}\frac{\partial{\bm{u}_k}(\tau_q)}{\partial{\theta_k}}\left(P_{j,k}^{+} - P_{j,k}^{-}\right),
\end{equation}
where
 \begin{equation}\label{eq:P}
 P_{j,k}^{\pm} = p(M_{H_j}=\pm1)\langle H_k\rangle_{\pm}.
 \end{equation}
and $\tau_p, \tau_q$ are uniformly sampled from $[0,T]$, $p(M_{H_j}=\pm 1)=\|\frac{1}{2}(I\pm H_j)|\psi(\tau_p)\rangle\|^2$ refers to the probability of getting the result $\pm 1$ via performing projective measurement $\mathcal{M}_{H_j}$ on $|\psi(\tau_p)\rangle$. $\langle H_k\rangle_{\pm}$ is the expectation value of $H_k$ over the state $|\psi(\tau_q)^{\pm}\rangle=U_{\bm{\theta}}(\tau_q,\tau_p)|\psi_{\pm}\rangle$ where $U_{\bm{\theta}}(\tau_q,\tau_p)$ is analog quantum evolution from $\tau_p$ to $\tau_q$ and $|\psi_{\pm}\rangle$ is the post projected state of $\mathcal{M}_{H_j}$. 
\end{theorem}

\begin{proof}
According to Theorem~\ref{thm:deravative_evo_opt}, $\frac{\partial{|\psi(\bm{\theta})\rangle}}{\partial{\theta_j}}$ can be expressed like,
\begin{equation}
\frac{\partial{|\psi(\bm{\theta})}\rangle}{\partial{\theta_j}} = -i\int_{0}^T
d\tau \frac{\partial{\bm{u}_j}}{\partial{\theta_j}} \cdot U_{\bm{\theta}}(T,\tau) H_j U_{\bm{\theta}}(\tau,0) |\psi_0\rangle.
\end{equation}
Thus, we have 
\begin{align}
M_{j,k}^R  &=  \Re\left[\frac{\partial\left<\psi(\bm{\theta})\right|}{\partial\theta_j}\frac{\partial\left|\psi(\bm{\theta})\right>}{\partial\theta_k}\right],\\
& = - \int_{0}^T\int_{0}^T d\tau d\tau' \frac{\partial{\bm{u}_j}(\tau)}{\partial{\theta_j}}\frac{\partial{\bm{u}_k}(\tau')}{\partial{\theta_k}} \cdot \Re\left[ \langle \psi_0| U_{\bm{\theta}}(0,\tau)H_j U_{\bm{\theta}}(\tau,\tau')H_k U_{\bm{\theta}}(\tau',0)|\psi_0\rangle \right],
\end{align}
Let $|\psi(\tau)\rangle = U_{\bm{\theta}}(\tau,0)|\psi_0\rangle$, since $\Re[P] = \Re[P^\dagger]$ where $P\in \mathbb{C}^{n\times n}$, $U_{\bm{\theta}}(\tau,\tau') = U_{\bm{\theta}}^{\dagger}(\tau',\tau)$ and Theorem~\ref{thm:measure}, thus we have
\begin{align}
M_{j,k}^R  &= - \int_{0}^T\int_{0}^T d\tau d\tau' \frac{\partial{\bm{u}_j}(\tau)}{\partial{\theta_j}}\frac{\partial{\bm{u}_k}(\tau')}{\partial{\theta_k}} \cdot \Re\left[ \langle \psi(\tau)|U_{\bm{\theta}}(\tau,\tau')H_k U(\tau',\tau)H_j|\psi(\tau)\rangle \right],\\
& = - \int_{0}^T\int_{0}^T d\tau d\tau' \frac{\partial{\bm{u}_j}(\tau)}{\partial{\theta_j}}\frac{\partial{\bm{u}_k}(\tau')}{\partial{\theta_k}} \cdot \left( P_{j,k}^{+} - P_{j,k}^{-}\right),
\end{align}
where 
\begin{equation}
P_{j,k}^{\pm} =  p(M_{H_j}=\pm 1)\langle H_k\rangle_{M_{H_j}=\pm 1}.
\end{equation}
To approximate the integration over $[0,T]$ via Monte Carlo method, we have
\begin{equation}
M_{j,k}^R \approx -\frac{T^2}{N^2}\sum_{p=1}^N\sum_{q=1}^N \frac{\partial{\bm{u}_j}(\tau_p)}{\partial{\theta_j}}\frac{\partial{\bm{u}_k}(\tau_q)}{\partial{\theta_k}}\cdot \left[p(M_{H_j}=+1)\langle H_k\rangle_{M_{H_j}=+1}-p(M_{H_j}=-1)\langle H_k\rangle_{M_{H_j}=-1}\right].
\end{equation}
\end{proof}

\begin{theorem}[{\bf{estimation of $V_{j}^R$}}]\label{thm:V_R}
 Given the target Hamiltonian $H$ and the controllable Hamiltonian $H(\bm{\theta},\tau)=\sum_{j=1}^m \bm{u}_j(\bm{\theta},\tau)H_j$, where $\{H_j\}_{j=1}^m$ are Pauli strings. Each $\tau_p$ is independently uniformly drawn from $[0,T]$. The real part of $V_{j}=\frac{\partial\left\langle\psi(\bm{\theta})\right|}{\partial\theta_{j}}H\left|\psi(\bm{\theta})\right\rangle$ can be estimate by
 \begin{equation}
 V_{j}^R \approx -\frac{T}{2N}\sum_{p=1}^N \frac{\partial{\bm{u}_j}(\tau_p)}{\partial{\theta_j}} \cdot \left(Q_{j}^{+}(\tau_p) - Q_{j}^{-}(\tau_p)\right), 
 \end{equation}
 where
 \begin{equation}
 Q_{j}^{\pm}(\tau_p) = \langle \psi^{\pm}|H|\psi^{\pm}\rangle.
 \end{equation}
 and
 \begin{equation}
 |\psi^{\pm}\rangle = U_{\bm{\theta}}(T,\tau_p) e^{-i(\pm\frac{1}{4})\pi H_j}U(\tau_p,0)|\psi_0\rangle.
 \end{equation}
\end{theorem}

\begin{proof}
According to Theorem~\ref{thm:deravative_evo_opt} and let $|\psi(\tau) \rangle = U_{\bm{\theta}}(\tau,0)|\psi_0\rangle$ thus we have,
\begin{align}
\frac{\partial\left\langle\psi(\bm{\theta})\right|}{\partial\theta_{j}}H|\psi(\bm{\theta})\rangle & = \int_{0}^T d\tau \frac{\partial \bm{u}_j(\tau)}{\partial \theta_j}\langle \psi(\tau)|(iH_j)U_{\bm{\theta}}(\tau,T)HU_{\bm{\theta}}(T,\tau)|\psi(\tau)\rangle
\end{align}
Let $A = \frac{\partial\left\langle\psi(\bm{\theta}_\tau)\right|}{\partial\theta_{j}}H|\psi(\bm{\theta}_\tau)\rangle $ and according to lemma \ref{lemma:phase_shift_rule}, we have
\begin{align}
\Re[A] &= \frac{1}{2}(A+A^\dagger)\\
& = \frac{1}{2}\int_{0}^T d\tau \frac{\partial \bm{u}_j(\tau)}{\partial \theta_j}\cdot \left(Q_{j}^{+}(\tau) - Q_{j}^{-}(\tau)\right)
\end{align}
where
\begin{equation}
Q_{j}^{\pm}(\tau) = \langle \psi(\tau)| e^{i(\pm\frac{1}{4})\pi H_j}U_{\bm{\theta}}(\tau,T)\cdot H \cdot U_{\bm{\theta}}(T,\tau) e^{-i(\pm\frac{1}{4})\pi H_j}|\psi(\tau)\rangle.
\end{equation}
Then, we can uniformally sample $\{\tau_p\}_{p=1}^N$ to approximate the integration and get the estimation of $V_j^R$ by
\begin{equation}
V_{j}^R  \approx -\frac{T}{2N}\sum_{p=1}^N \frac{\partial{\bm{u}_j}(\tau_p)}{\partial{\theta_j}} \cdot (Q_{j}^{+}(\tau_p)) - Q_{j}^{-}(\tau_p)).
\end{equation}
\end{proof}

\begin{theorem}[estimation of $V_{j}^I$]\label{thm:V_I}
Given the problem Hamiltonian $H$ and the controllable Hamiltonian $H(\bm{\theta}, \tau)=\sum_{j=1}^m \bm{u}_j(\bm{\theta},\tau)H_j$, where $\{H_j\}_{j=1}^m$ are Pauli strings. Each $\tau_p$ is independently and  uniformally sampled from $[0,T]$. The imaginary part of $V_{j}=\frac{\partial\left\langle\psi(\bm{\theta})\right|}{\partial\theta_{j}}H\left|\psi(\bm{\theta})\right\rangle$ can be estimated by
\begin{equation}
V_{j}^I \approx -\frac{T}{N}\sum_{p=1}^N \frac{\partial{\bm{u}_j}(\tau_p)}{\partial{\theta_j}} \cdot \left[p(M_{H_j}=+1)\langle H \rangle_{M_{H_j}=+1}-p(M_{H_j}=-1)\langle H \rangle_{M_{H_j}=-1}\right], 
\end{equation}

\end{theorem}

\begin{proof}
According to the fact that $\mathcal{I}[-i S] = - \mathcal{R}[S]$ for arbitrary complex matrix $S$ and theorem \ref{thm:measure}, we are also capable of utilizing Monte Carlo integration method via independently and uniformly drawing samples $\{\tau_p\}_{p=1}^N$ from $[0,T]$ to estimate $V_j^I$.
\end{proof}

\IncMargin{1em}
\begin{algorithm}
\SetKwData{Left}{left}\SetKwData{This}{this}\SetKwData{Up}{up}
\SetKwFunction{Union}{Union}\SetKwFunction{FindCompress}{FindCompress}
\SetKwInOut{Input}{input}\SetKwInOut{Output}{output}

\Input{Number of sampling $N$, controllable Hamiltonian $H(\bm{\theta},\tau)= \sum_{j=1}^m \bm{u}_j(\bm{\theta},\tau)H_j$ with $m$ control pulses, parameters of control pulses $\bm{\theta}$, evolution duration $T$, initial quantum state $|\psi_0\rangle$, index $j,k$.}
\Output{the estimation of $M_{j,k}^R$}
\BlankLine
\For{$p \leftarrow 1$ \KwTo $N$}
{
	\For{$q \leftarrow 1$ \KwTo $N$}
	{
		Independently sampling $\tau_p,\tau_q$, where $\tau_p < \tau_q$, from uniform distribution $U(0,T)$\;
		Evolve the initial state $|\psi_0\rangle$ from $0\rightarrow \tau_p$ to get state $|\psi(\tau_p)\rangle$\;
		Perform nondestructively projective measurement $\mathcal{M}_{H_j}$ on $|\psi(\tau_p)\rangle$, get $p(M_{H_j}=\pm 1)$ and associated post-state $|\psi_{\pm}\rangle$\;
		Evaluate state $|\psi_{\pm}\rangle$ from $\tau_p\rightarrow \tau_q$ to get states $|\psi(\tau_q)^{\pm}\rangle$\;
		Estimate the expectation value of $H_k$ over states $|\psi(\tau_q)^{\pm}\rangle$, i.e. $\langle H_k \rangle_{\pm}$ respectively.\;
		Calculate the $P_{j,k}^{\pm} = p(M_{H_j}=\pm1)\langle H_k\rangle_{\pm}$.
	}
}
$\tilde{M}_{j,k}^R \leftarrow -\frac{T^2}{N^2}\sum_{p=1}^N\sum_{q=1}^N \frac{\partial{\bm{u}_j}(\tau_p)}{\partial{\theta_j}}\frac{\partial{u_k}(\tau_q)}{\partial{\theta_k}}\cdot (P_{j,k}^{+} - P_{j,k}^{-})$.
\caption{Estimating matrix $M_{j,k}^R$}\label{algo_MR}\DecMargin{1em}
\end{algorithm}\DecMargin{1em}

\IncMargin{1em}
\begin{algorithm}[H]
\SetKwData{Left}{left}\SetKwData{This}{this}\SetKwData{Up}{up}
\SetKwFunction{Union}{Union}\SetKwFunction{FindCompress}{FindCompress}
\SetKwInOut{Input}{input}\SetKwInOut{Output}{output}

\Input{Number of sampling $N$, controllable Hamiltonian $H(\bm{\theta},\tau)= \sum_{j=1}^m \bm{u}_j(\bm{\theta},\tau)H_j$ with $m$ control pulses, parameters of control pulses $\bm{\theta}$, evolution duration $T$, initial quantum state $|\psi_0\rangle$, index $j,k$.}
\Output{the estimation of $V_{j}^R$}
\BlankLine
\For{$p \leftarrow 1$ \KwTo $N$}
{
		Independently sampling $\tau_p$ from uniform distribution $U(0,T)$\;
		Evolve the initial state $|\psi_0\rangle$ from $0\rightarrow t_p$ to get state $|\psi(\tau_p)\rangle$\;
		Perform $e^{-i(\pm\frac{1}{4})\pi H_j}$ on $|\psi(\tau_p)\rangle$\;
		Then perform $U_{\bm{\theta}}(T,\tau_p)$ to get $|\psi^{\pm}\rangle$\;
		Estimate the expectation value of $H$ over states $|\psi^{\pm}\rangle$, i.e. $Q_j^{\pm}$ respectively.
}
$\tilde{V}_{j}^R \leftarrow -\frac{T}{2N}\sum_{p=1}^N \frac{\partial{\bm{u}_j}(\tau_p)}{\partial{\theta_j}} \cdot \left(Q_{j}^{+}(\tau_p)) - Q_{j}^{-}(\tau_p)\right)$\;
\caption{Estimating matrix $V_{j}^R$}\label{algo_VR}\DecMargin{1em}
\end{algorithm}\DecMargin{1em}

\IncMargin{1em}
\begin{algorithm}
\SetKwData{Left}{left}\SetKwData{This}{this}\SetKwData{Up}{up}
\SetKwFunction{Union}{Union}\SetKwFunction{FindCompress}{FindCompress}
\SetKwInOut{Input}{input}\SetKwInOut{Output}{output}
\Input{Number of sampling $N$, controllable Hamiltonian $H(\bm{\theta},\tau)= \sum_{j=1}^m \bm{u}_j(\bm{\theta},\tau)H_j$ with $m$ control pulses, parameters of control pulses $\bm{\theta}$, evolution duration $T$, initial quantum state $|\psi_0\rangle$, index $j,k$.}
\Output{the estimation of $V_{j}^I$}
\BlankLine
	\For{$p \leftarrow 1$ \KwTo $N$}
	{
		Independently sampling $\tau_p$ from uniform distribution $U(0,T)$\;
		Evolve the initial state $|\psi_0\rangle$ from $0\rightarrow \tau_p$ to get state $|\psi(\tau_p)\rangle$\;
		Perform nondestructively projective measurement $\mathcal{M}_{H_j}$ on $|\psi(\tau_p)\rangle$, get $p(M_{H_j}=\pm 1)$ and associated post-state $|\psi_{\pm}\rangle$\;
		Estimate the expectation value of $H$ over states $|\psi(\tau_p)^{\pm}\rangle$, i.e. $\langle H \rangle_{\pm}$\;
		Calculate the $Q_{j}^{\pm}(\tau_p) = p(M_{H_j}=\pm1)\langle H\rangle_{\pm}$.
	}
$\tilde{V}_{j}^I \leftarrow -\frac{T}{2N}\sum_{p=1}^N \frac{\partial{\bm{u}_j}(\tau_p)}{\partial{\theta_j}} \cdot \left(Q_{j}^{+}(\tau_p)) - Q_{j}^{-}(\tau_p)\right)$
\caption{Estimating matrix $V_{j}^I$}\label{algo_VR}\DecMargin{1em}
\end{algorithm}\DecMargin{1em}

\section{Expressivity and entanglement capability of parametrized analog Hamiltonian}

\subsection{Expressivity}
In the context of variational quantum algorithm, the expressivity characterizes the capability of an ansatz to represent a broad class of quantum states or transformations. Specifically, it quantifies how extensively the circuit can explore the Hilbert space as its tunable parameters are varied. A highly expressive ansatz can approximate a large and diverse set of quantum states. Here, we follow the definition of the expressivity proposed in Ref.~\cite{sim2019expressibility} for parametrized analog Hamiltonian.

\begin{definition}[Expressibility of a Variational Quantum Ansatz in Ref.~\cite{sim2019expressibility}]
Let $\mathcal{U}(\bm{\theta})$ be a variational quantum ansatz with $p$ parameters. The ansatz induces a distribution $\hat{P}_{\mathcal{U}}(F)=p(F = |\langle \psi(\bm{\theta})| \psi(\bm{\theta}') \rangle|^2)$ over the fidelity $F = |\langle \psi(\bm{\theta})| \psi(\bm{\theta}') \rangle|^2$
where $\bm{\theta}$ and $\bm{\theta}'$ are independently and uniformally drawn from $\mathbb{R}^p$ and $|\psi(\bm{\theta})\rangle = \mathcal{U}(\bm{\theta}) |0\rangle$.
Let $P_{\mathrm{Haar}}(F)$ denote the fidelity distribution obtained by sampling pure states uniformly at random according to the Haar measure on $\mathcal{H} = (\mathbb{C}^2)^{\otimes n}$.  
The expressibility $\mathcal{E}(\mathcal{U})$ of the variational quantum ansatz is defined as the Kullback-Leibler (KL) divergence between $\hat{P}_{\mathcal{U}}(F)$ and $P_{\mathrm{Haar}}(F)$,
\begin{equation}
\mathcal{E}(\mathcal{U})=D_{\mathrm{KL}}\!\left(
\hat{P}_{\mathcal{U}}(F)\middle\|P_{\mathrm{Haar}}(F)\right).
\end{equation}
where $P_{\mathrm{Haar}}(F)$ can be analytically estimated by
\begin{equation}
    P_{\mathrm{Haar}}(F)= (N-1)(1-F)^{N-2}, N=2^n
\end{equation}
\end{definition}
A smaller value of $\mathcal{E}(\mathcal{U})$ signifies that the ensemble of states generated by the ansatz more closely approximates the Haar measure, indicating higher expressibility.

\subsection{Entanglement Capability}
The entanglement capability of a variational quantum ansatz is employed to quantify its ability to generate multipartite quantum correlations, which are essential for the expressive power and potential advantage of variational quantum algorithms. The entanglement capability can be characterized using the Meyer–Wallach global entanglement measure, which evaluates the average linear entropy of single-qubit reduced density matrices for a given pure state. This measure captures how strongly each individual qubit is entangled with the rest of the system and thus provides a meaningful indicator of global entanglement rather than pairwise correlations alone \cite{meyer2001global}. The entangling capability of an ansatz is then defined as the expectation value of the Meyer–Wallach measure over randomly sampled variational parameters that reflects the typical entanglement generated by the ansatz. 

\begin{definition}[Entangling Capability (Meyer--Wallach Measure)]
Let $\mathcal{U}(\bm{\theta})$ be a variational quantum ansatz with parameters $\bm{\theta}$ independently and uniformally drawn from $\mathbb{R}^p$. The ansatz prepares a pure quantum state
$|\psi(\bm{\theta})\rangle = \mathcal{U}(\bm{\theta}) |0\rangle$.
For qubit $j \in \{1,\dots,n\}$, the reduced density matrix $\rho_j(\bm{\theta})$ is denote by $ \rho_j(\bm{\theta}) = \operatorname{Tr}_{n\backslash j}\left[|\psi(\bm{\theta})\rangle\langle\psi(\bm{\theta})|\right]$ obtained by tracing out all qubits except $j$.  
The Meyer-Wallach global entanglement measure of the state, which quantifies the average linear entropy of the single-qubit reduced states, is defined as
\begin{equation}
Q(\bm{\theta})=\frac{2}{n}\sum_{j=1}^{n}\left(1 - \operatorname{Tr}\left[\rho_j(\bm{\theta})^2\right]
\right).
\end{equation}
The entangling capability of the variational quantum ansatz is defined as the expectation value over random variational parameters, 
\begin{equation}
\mathcal{C}_{\mathrm{ENT}}(\mathcal{U})= \mathbb{E}_{\bm{\theta} \sim P_{\bm{\theta}}}\left[\, Q(\bm{\theta}) \,\right].
\end{equation}
\end{definition}
The value of such entanglement capability is between $0$ and $1$. A higher entanglement capability indicates that the circuit more readily explores highly entangled regions of the Hilbert space, while limited entangling capability may restrict the reachable state manifold and impair the representational power of the ansatz.

\subsection{Analysis on the analog Hamiltonian}
In this section, we analyze the advantages of the parametrized analog Hamiltonian in terms of both expressibility and entangling capability. To this end, we compare the expressibility and entanglement capability of six different four-qubit ansatzes, including both analog and gate-based circuit constructions.
Ansatz 1 and ansatz 2 share the same underlying structure as the Hamiltonian, 
\begin{equation}
    H(t) = \sum_{v,i}f_{v,i}(t)\sigma_{v,i} + \sum_{i,j}h_{ij}(t)(\sigma_i^{+}\sigma_j^{-}+\sigma_i^{-}\sigma_j^{+})
\label{eq:H_transmon}
\end{equation}
They are composed of the control pulses $\sigma_{x_i}$, $\sigma_{z_i}$, and $(\sigma_i^{+}\sigma_j^{-} + \sigma_i^{-}\sigma_j^{+})$. For simplicity, the maximum amplitude of each control pulse is fixed to unity. The initial state is chosen as $|i\rangle^{\otimes 4}$, where $|i\rangle$ denotes the eigenstate of $\sigma_y$ with eigenvalue $+1$. The total evolution time of ansatz~1 is set to $2\pi$, whereas for ansatz~2 it is extended to $2n\pi$, with $n$ denoting the number of parameters controlling each pulse.
Ansatz 3 and 4 employ the same set of control pulses as ansatz 1 and 2, respectively, but differ in their temporal structure. The non-commuting control Hamiltonians are applied sequentially in time rather than simultaneously. Finally, ansatz 5 and ansatz 6 are gate-based and composed of standard CNOT and controlled-rotation gates. To ensure a fair comparison, the number of tunable parameters per layer in ansatz 5 and 6 is chosen to be comparable to the number of control pulses in ansatz 1 and 2. For each of the six ansatzes, we separately evaluate their expressibility and entanglement capability.

\begin{figure}[H]
    \centering
    \setkeys{Gin}{width=0.5\textwidth}
\subfloat[Ansatz 5
          \label{fig:ansatz5}]{\scalebox{1}{
            \begin{quantikz}
                \lstick{\( \ket{0} \)} & \gate{R_y} & \gate{R_z} & \ctrl{1} & \qw & \qw & \qw &\qw\\
                \lstick{\( \ket{0} \)} & \gate{R_y} & \gate{R_z} & \targ{}  & \gate{R_y} & \gate{R_z} & \ctrl{1} &\qw\\
                \lstick{\( \ket{0} \)} & \gate{R_y} & \gate{R_z} & \ctrl{1} & \gate{R_y} & \gate{R_z} & \targ{}  &\qw \\
                \lstick{\( \ket{0} \)} & \gate{R_y} & \gate{R_z} & \targ{}  & \qw     & \qw        & \qw  &\qw 
            \end{quantikz}
        }}
    \hfill
\subfloat[Ansatz 6
          \label{fig:ansatz6}]{\scalebox{1}{
            \begin{quantikz}
                \lstick{\( \ket{0} \)} & \gate{R_x} & \gate{R_z} & \ctrl{1} & \qw & \qw  \\
                \lstick{\( \ket{0} \)} & \gate{R_x} & \gate{R_z} & \gate{R_x}  & \ctrl{1} & \qw  \\
                \lstick{\( \ket{0} \)} & \gate{R_x} & \gate{R_z} & \ctrl{1} & \gate{R_x} & \qw   \\
                \lstick{\( \ket{0} \)} & \gate{R_x} & \gate{R_z} & \gate{R_x}  & \qw     & \qw    
            \end{quantikz}
        }}

\caption{\small{(a) and (b) are gate-based ansatz used to be the benchmark with variationl analog Hamiltonian.}}
\label{fig:common}
    \end{figure}
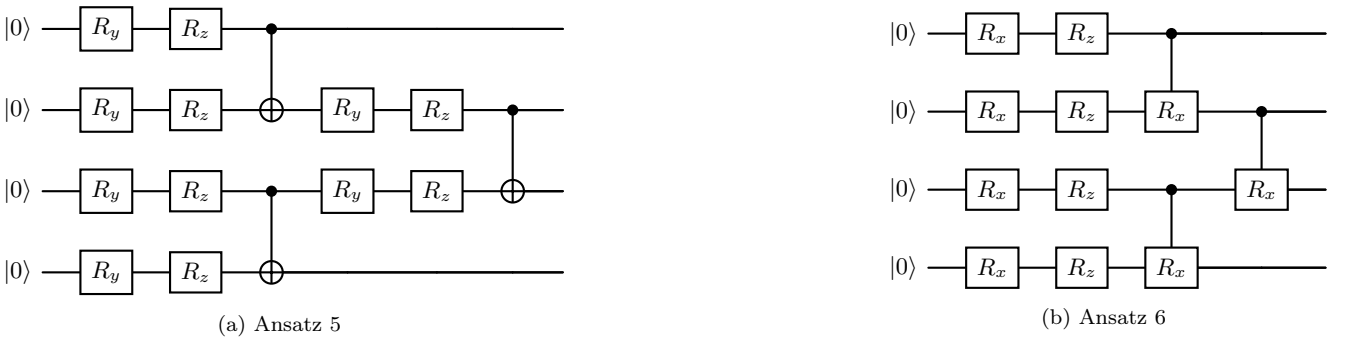

As a benchmark, we first evaluate the expressibility of two representative gate-based ansatzes, namely ansatz~5 and ansatz~6, by varying the number of layers from $1$ to $5$. In parallel, we investigate the expressibility of the analog Hamiltonian by increasing the number of control parameters per pulse also from $1$ to $5$. 
Since the analytic form of the fidelity distribution between random quantum states is unknown which makes direct estimation of the KL-divergence between the fidelity distributions is nontrivial. To enable a fair and consistent comparison across different ansatz, we adopt a uniform numerical procedure. For each ansatz, we randomly sample 40,000 pairs of parametrized quantum states, compute and sort their fidelities, and partition the data into 200 intervals. Within each interval, we estimate the probability density of the corresponding Haar-random fidelity distribution and subsequently evaluate the KL divergence. Similarly, we randomly selected 10,000 parameterized quantum states to calculate their average Meyer-Wallach measure. The expressibility and entanglement capacity comparison are presented in Fig.~\ref{fig:dkl} and Fig.~\ref{fig:ent}, respectively.

\begin{figure}[H]
    \centering
    \includegraphics[width=0.8\textwidth]{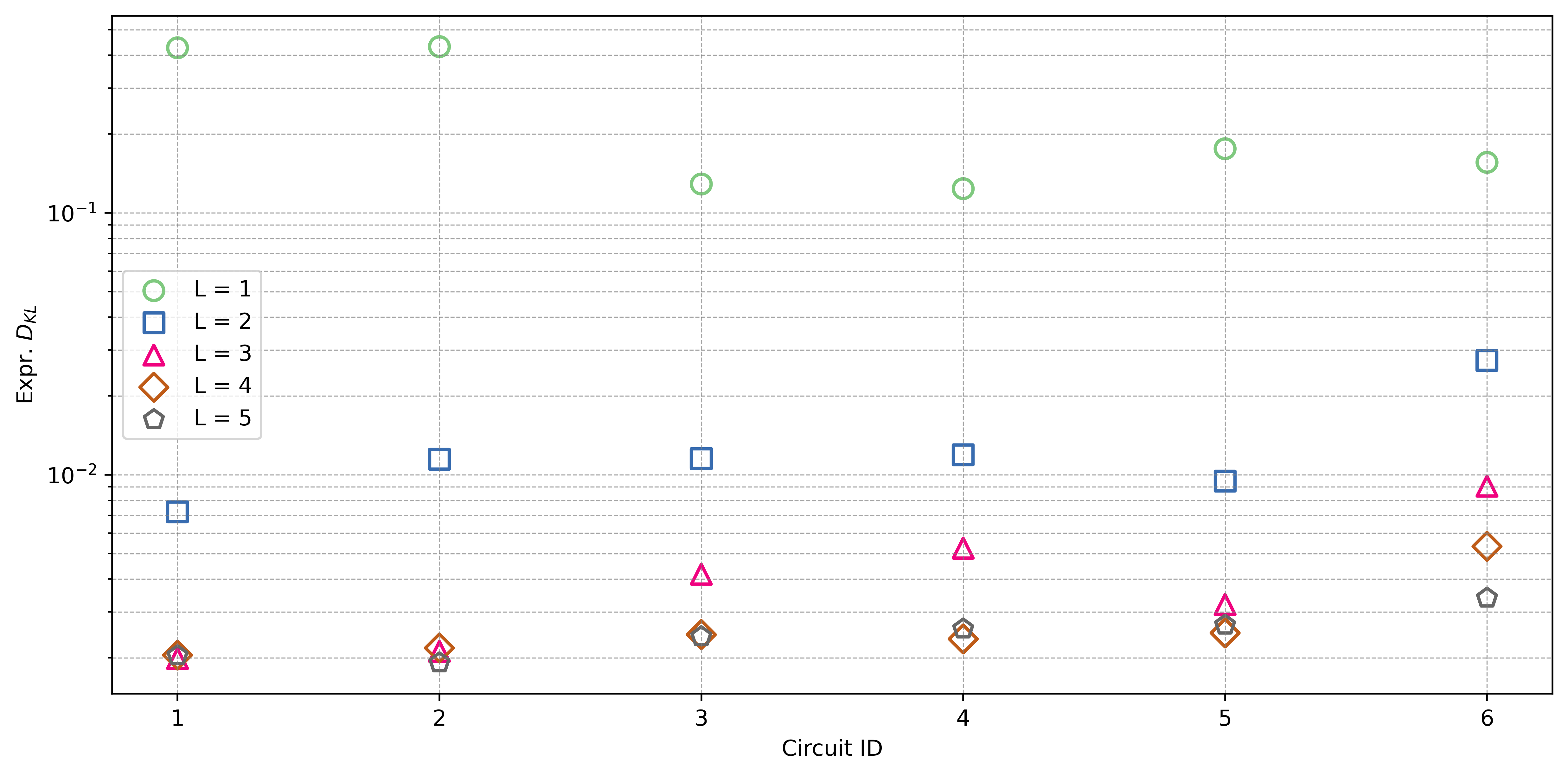}
    \caption{The KL divergence between the distribution of state fidelities of the ansatz 1-6 and the Haar random states. Ansatz with smaller KL divergence indicates higher expressivity.}
    \label{fig:dkl}
\end{figure}

\begin{figure}[H]
    \centering
    \includegraphics[width=0.8\textwidth]{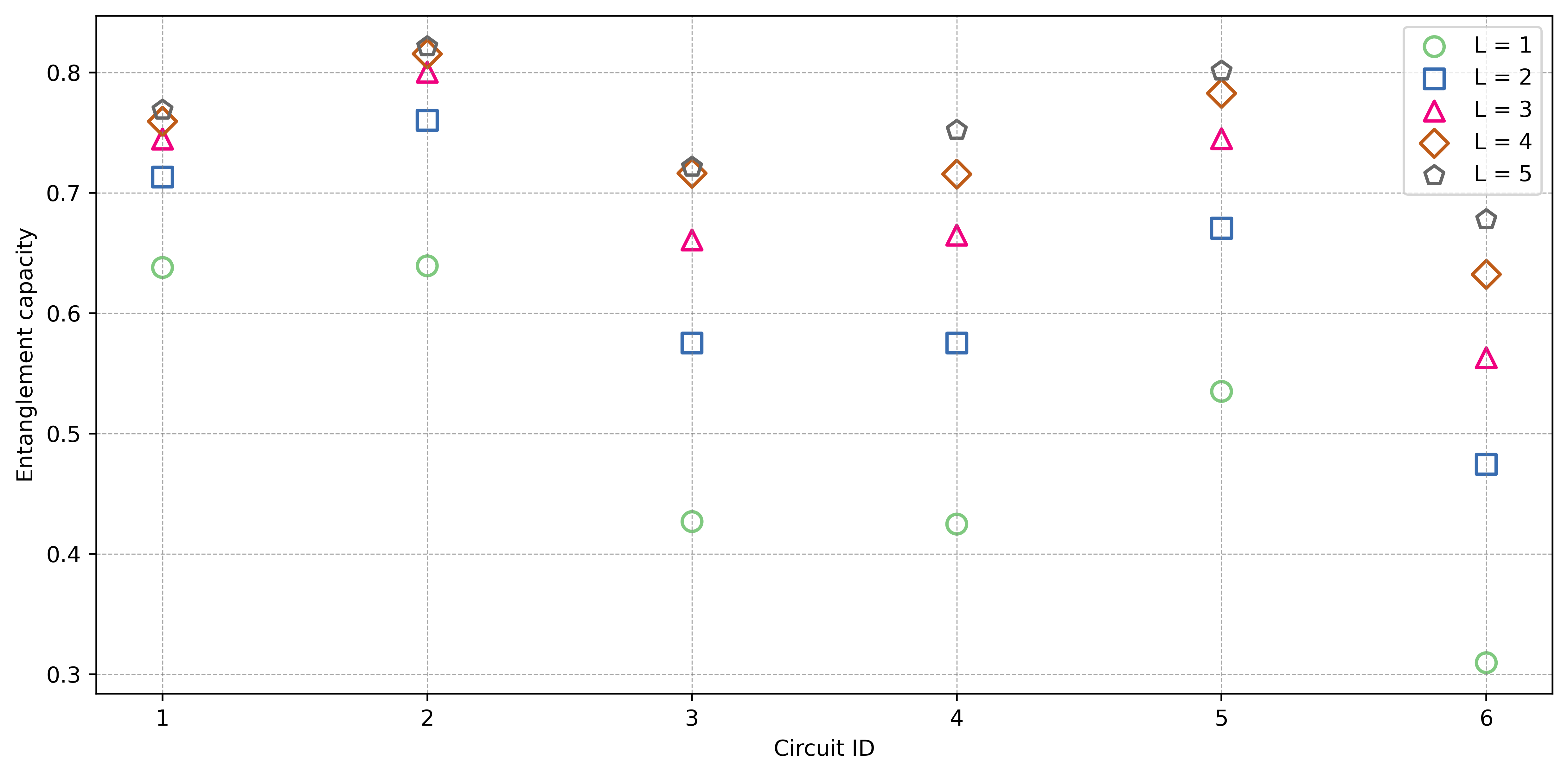}
    \caption{The average Meyer-Wallach Q-measure of the ansatz 1-6. Ansatz with higher Meyer-Wallach measure has stronger entangling capability.}
    \label{fig:ent}
\end{figure}

From the results shown in the two figures above, we have the following observations. 
First, when the number of control parameters per pulse exceeds two, ansatz~1 and~2 consistently exhibit higher expressibility and entangling capability than ansatz~3 and~4, which indicates that the simultaneous application of control pulses, rather than their sequential separation in time, substantially enhances both the expressibility and the entangling capability of the circuit. 
Second, a comparison between ansatz~1 and the gate-based ansatz~5 and~6 reveals distinct advantages of the analog Hamiltonian. In terms of expressibility, once the number of control parameters per pulse exceeds two, ansatz~1 surpasses ansatz~5 and~6 when they have the same number of layers. With respect to entangling capability, ansatz~1 is slightly weaker than ansatz~5 for the same evolution time. However, by extending the total evolution time, the entangling capability of ansatz~1 can exceed that of ansatz~5.
Finally, we compare ansatz~1 and ansatz~2 to examine the role of evolution time. In terms of expressibility, for $L = 3, 4,$ and $5$, the two ansatzes exhibit comparable performance despite their different total evolution times. In contrast, ansatz~2 consistently achieves a higher entangling capability than ansatz~1, indicating that longer evolution times can further enhance the ansatz ability to generate multipartite entanglement. To systematically characterize these trends, we further analyze the dependence of expressibility and entangling capability on both the number of control parameters per pulse and the total evolution time, with the results summarized in Fig.~\ref{fig:exp-ent}.

\begin{figure}[H]
    \centering
    \begin{minipage}{0.45\textwidth}
        \includegraphics[width=1\textwidth]{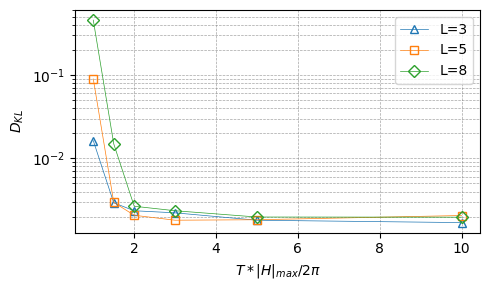}
    \end{minipage}
    \hfill
    \begin{minipage}{0.45\textwidth}
        \includegraphics[width=1\textwidth]{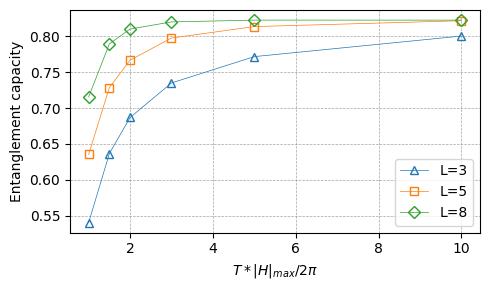}
    \end{minipage}
    \caption{The expressivity(left) and entanglement(right) capability of the ansatz 1 with different number of control parameters per pulse and total evolution time. L refer to the number of control parameters per pulse.$\|H\|_{max}$ refers to the maximum amplitude of the each pulses.}
    \label{fig:exp-ent}
\end{figure}

\section{Complexity and error analysis}
In this section, we analyze the computational complexity and error of the UAQS framework. We first examine the complexity of estimating the element of $M$ and $V$ in UAQS. We then characterize the dominant sources of error arising in UAQS, such as algorithm error and implementation error, and discuss how these errors propagate through the optimization process.
 
\subsection{complexity analysis} Here, we first discuss the computational resource estimation of the proposed UAQS. Let $m_{H_{\bm{\theta}}}$ be the number of terms of controllable Hamiltonian $H_c$, $m_H$ be the terms of problem Hamiltonian $H$, $N$ be the sampling number in Monte Carlo integration, $m_s$ be associated with quantum process corresponding to the shot number for estimating $p(M_{H_j}=\pm1)$. The complexity of estimating all the $M_{j,k}^R$ via Theorem~\ref{thm:M_R} scales with $\mathcal{O}(N^2 m_{H_{\bm{\theta}}}^2 m_s T^2)$.  Similarly, $V_{j}^I$ takes $\mathcal{O}(N m_{H}m_{H_{\bm{\theta}}}m_s T)$ time to estimate. Being different from calculating $V_j^I$ which involves projective measurement $\mathcal{M}_H$, all the $V_j^R$ can be directly estimated via parameter shift rule which will cost $\mathcal{O}(N m_H m_{H_{\bm{\theta}}}T)$. 

\subsection{error analysis}
To estimate the errors occurred in UAQS, we utilize the trace distance as the metric, i.e. $D(|\psi\rangle,|\psi'\rangle)=\sqrt{1 - |\langle\psi'|\psi\rangle|^2}$. The overall errors occurred during the simulation can be bounded by algorithm error $\mathcal{E}_a$ and device error $\mathcal{E}_i$, which is illustrated in Fig.~\ref{fig:error_illus}. Let $|\psi(\bm{\theta},t_N)\rangle$ and $|\phi_{t_N}\rangle$ be the trial states and true state at time $t_N$, respectively. Thus, we have
\begin{equation}
D(|\psi(\bm{\theta},t_N)\rangle,|\phi_{t_N}\rangle) \le \mathcal{E}_0 + \mathcal{E}_a + \mathcal{E}_d 
\end{equation}
where $\mathcal{E}_a$ and $\mathcal{E}_i$ are defined as
\begin{align}
\mathcal{E}_0 & = D(|\psi_0\rangle, |\phi_0\rangle)\\
\mathcal{E}_a &= \sum_{n}^{N}\mathcal{E}_a^{(n)} = \sum_{n=1}^{N} D(|\psi(\bm{\theta},t_n)\rangle,U_{n-1}|\psi(\bm{\theta},t_{n-1})\rangle)\\
\mathcal{E}_i &= \sum_{n}^{N}\mathcal{E}_i^{(n)} = \sum_{n=1}^{N} D(|\psi'(\bm{\theta},t_n)\rangle, |\psi(\bm{\theta},t_n)\rangle)
\end{align}
where $\sum_{n} t_n = T$. For ease the notation, we omit the parameters of trial state, denote $|\psi(\bm{\theta},t)\rangle = |\psi_t\rangle$ as the idea generated trial state, and define $|\psi_t'\rangle$ as the trial states implemented on real quantum device.

\begin{figure}[H]
    \centering
	\includegraphics[width=0.5\textwidth]{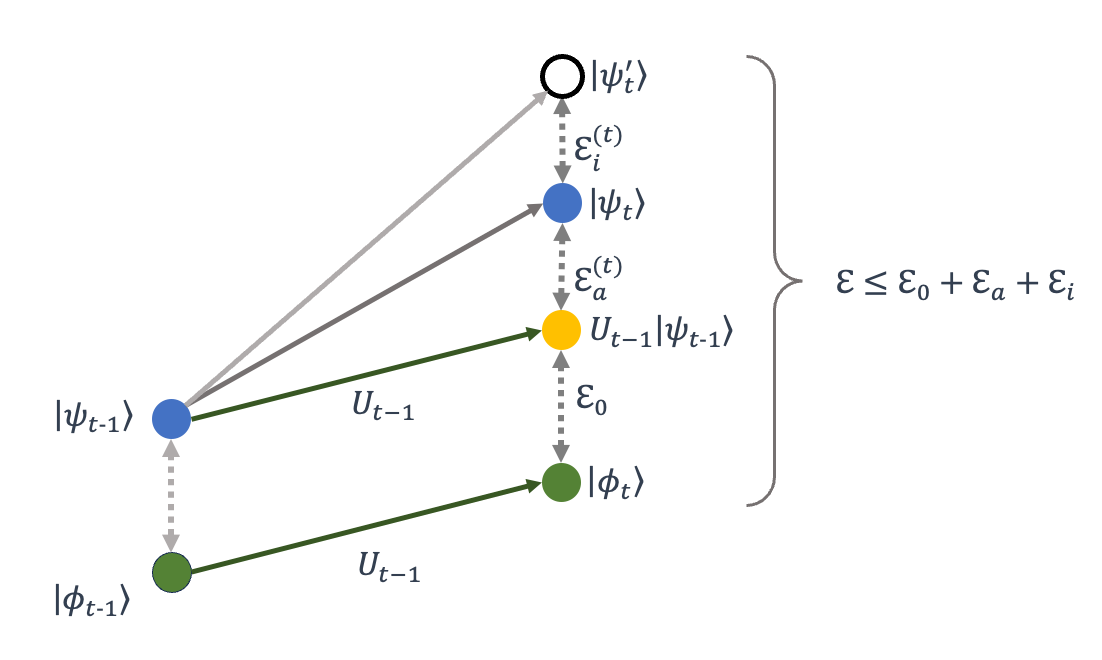}
\caption{The illustration of states at evolution time $t$. The green circles represents the state $|\phi_t\rangle$ which is under the exact evolution of hamiltonian $H$, the blue circles denotes ideal trial state  $|\psi_t\rangle$ which is generated by the proposed algorithm, the yellow circle refers to the trial state $|\psi_{t-1}\rangle$ evolved by the exact evolution $U_{t-1}$, and the white circle is the trial states prepared on real quantum device $|\psi_t'\rangle$.}
\label{fig:error_illus}
\end{figure}

\paragraph{algorithm error $\mathcal{E}_a$}
First, we consider the algorithm error in time step $t$, the algorithm error $\mathcal{E}_a$ is defined as the difference between parameter trial states prepared under the exact time evolution and the proposed algorithm. To analysis such error, we use Taylor expansion of exact evolution during $\delta t$, $U' = e^{-iH\delta t}$, and the parametrized analog Hamiltonian $U(\bm{\theta}(t))$, then we can get,
\begin{align}\label{eq:expan_u'}
U' &= e^{-iH\delta t} \\
&= \sum_{k=0}^\infty \frac{\delta t^k}{k!}(-i H)^k\\
& = \mathbb{I} + \delta t (-i H) + \frac{\delta t^2}{2!}(-i H )^2+\mathcal{O}(\delta t^3)
\end{align}
and
\begin{align}\label{eq:expan_u}
U(\bm{\theta}_t) &= U(\bm{\theta}_{t-1}+ \delta t \cdot \dot{\bm{\theta}}_{t-1})\\
& = \sum_{k=0}^\infty \frac{\delta t^k}{k!}\frac{\partial^k U(\bm{\theta}_{t-1})}{\partial t^k}\\
& = U(\bm{\theta}_{t-1})+\delta t  \frac{U(\bm{\theta}_{t-1})}{\partial t} + \frac{\delta t^2}{2!}\frac{\partial^2 U(\bm{\theta}_{t-1})}{\partial t^2}+\mathcal{O}(\delta t^3)
\end{align}

Here, the $\mathcal{E}_a$ during time step $t$ can be written as
\begin{align}
\mathcal{E}_a &= D(e^{-iH\delta t}|\psi_{t-1}\rangle, U(\bm{\theta}_t)|0\rangle)\\
& = \sqrt{1 - |\langle \psi_{t-1}| e^{iH\delta t}\cdot U(\bm{\theta}(t))|0\rangle|^2}
\end{align}
When we substitute Eq.(\ref{eq:expan_u}) and Eq.(\ref{eq:expan_u'}) into $\mathcal{E}_a$, we have
\begin{equation}
    \mathcal{E}_a = \sqrt{\Delta \delta t^2 + \mathcal{O}(\delta t^3)}
\end{equation}
as $\Re\left(\langle \psi_{t-1}|\frac{\partial^2 U(\bm{\theta}_{t-1})}{\partial t^2}|0\rangle\right)$ + $\langle 0|\frac{\partial U^\dagger(\bm{\theta}_{t-1})}{\partial t}\frac{\partial U(\bm{\theta}_{t-1})}{\partial t}|0\rangle = 0$ and $\Re\left(\langle \psi_{t-1}|\frac{\partial U(\bm{\theta}_{t-1})}{\partial t}|0\rangle\right) = 0$, we can get,
\begin{align}
    &\Delta = \langle \psi_{t-1}|H|\psi_{t-1}\rangle + \left|\langle \psi_{t-1}|H|\psi_{t-1}\rangle\right|^2 + \left|\langle \psi_{t-1}|\frac{\partial U(\bm{\theta}_{t-1})}{\partial t}|0\rangle \right|^2 \nonumber \\
    &+ \Re\left(\langle\psi_{t-1}|(-iH)\frac{\partial U(\bm{\theta}_{t-1})}{\partial t}|0\rangle\right) 
    + \Re\left(\langle \psi_{t-1}|\frac{\partial U(\bm{\theta}_{t-1})}{\partial t}|0\rangle \langle \psi_{t-1}|(-iH)|\psi_{t-1}\rangle\right) 
\end{align}
Thus, we can bound the norm of $\Delta$,
\begin{align}
    \|\Delta\| \leq \|H\|+ \|H\|^2 + \left\|\frac{\partial U(\bm{\theta}_{t-1})}{\partial t}\right\|^2+2\left\|H\right\|\left\|\frac{\partial U(\bm{\theta}_{t-1})}{\partial t}\right\|
\end{align}

\paragraph{implementation error $\mathcal{E}_i$}
The implementation error is mainly caused by the imperfection estimations of $M$ and $V$ affected by either device noise or quantum measurement. Thus, we denote $M = \hat{M} + \delta M$ where $M$ is accurate and $\tilde{M}$ is the estimation. Similarly, for $V$, we have $V = \tilde{V} + \delta V$. Therefore, the error of the derivative of parameters $\dot{\bm{\theta}}$ can be expressed as,
\begin{align}
    \delta \dot{\bm{\theta}} & = M^{-1}V - \hat{M}^{-1}\hat{V}\\
    & =  (\hat{M} + \delta M)^{-1}(\hat{V}+\delta V) - \hat{M}^{-1}\hat{V}\\
    & \approx \hat{M}^{-1}\delta V - \hat{M}^{-1}\delta M \hat{M}^{-1} V.
\end{align}
As the Ref. \cite{li2017efficient} has pointed out the implementation error in each time step $n$ can be written as
\begin{equation}
\mathcal{E}_i^{(n)} = \sqrt{\delta \dot{\bm{\theta}} A \delta \dot{\bm{\theta}} \delta t^2 + O(\delta t^3)},
\end{equation}
where $A$ is the  quantum geometric tensor,
\begin{equation}
A_{j,k} = \frac{\partial{\langle \psi(\bm{\theta}(t))|}}{\partial \theta_j}\frac{\partial{|{\langle \psi(\bm{\theta}(t))\rangle}}}{\partial \theta_k} -\frac{\partial{\langle \psi(\bm{\theta}(t))|}}{\partial \theta_j} |\psi(\bm{\theta}(t))\rangle \langle \psi(\bm{\theta}(t))| \frac{\partial{|{\langle \psi(\bm{\theta}(t))\rangle}}}{\partial \theta_k} 
\end{equation}
Therefore the total implementation error $\mathcal{E}_i$ can be bounded by
\begin{equation}
\mathcal{E}_i \leq \sqrt{\|A\|_{\max}}\|\Delta\|_{\max} T,
\end{equation}
where
\begin{equation}
\Delta = \|\hat{M}\|^{-1}\|\delta V\| + \| \hat{M}^{-1}\|^2 \|\delta M\| \|V\|
\end{equation}

\section{Supplementary numerical result}
In this supplementary section, we present additional numerical experiments that further demonstrate the applicability and robustness of the proposed universal analog quantum simulation (UAQS) framework. 
In particular, we focus on two representative and practically relevant applications. The first concerns adiabatic state preparation, where UAQS is employed to track the ground state of a time-dependent Hamiltonian and the second addresses ground-state preparation for molecular systems, highlighting the ability of UAQS to capture electronic structure properties in quantum chemistry models. Together, these case studies provide complementary evidence of the effectiveness of UAQS across both dynamical and static quantum simulation tasks, and underscore its versatility in addressing problems of interest in condensed matter physics and molecular quantum chemistry.
\subsection{Adiabatic state preparation}
Adiabatic state preparation (ASP) is a widely used approach for preparing the ground state of a complex Hamiltonian. The central idea is to initialize the quantum system in the ground state of a simple and easily implementable Hamiltonian, denoted as $H_{\operatorname{start}}$, and then gradually transform the system Hamiltonian into the target Hamiltonian  $H_{\operatorname{final}}$.
This process is realized through a time-dependent Hamiltonian $H(t)$ that interpolates between these two Hamiltonian,
\begin{equation}
H(t) = \left[1 - s(t)\right] H_{\text{start}} + s(t) H_{\text{final}},
\end{equation}
where $s(t) \in [0,1]$ is a monotonically increasing schedule, often chosen to be linear, such as $s(t) = t/T$, with total evolution time $T$.

In the framework, the evolution is discretized into small iteration steps. When the step size is sufficiently small, the Hamiltonian can be regarded as approximately constant within each iteration. To illustrate this procedure, we consider the following transverse-field Ising model,
\begin{equation}
H = \sum_{i=1}^{n} Z_i Z_{i+1} - \sum_{i=1}^{n} X_i.
\end{equation}
Specifically, we choose the initial Hamiltonian as $H_{\text{start}} = \sum_{i=1}^{n} X_i$, where $n=6$ and whose ground state $|\psi_0\rangle = |+\rangle^{\otimes n}$ is straightforward to prepare. Starting from this state, we apply the UAQS algorithm with a total evolution time $T = 5$ and a step size $\delta t=0.01$ to approximate the ground state of the target Hamiltonian. The resulting evolution of the energy and the fidelity with respect to the target ground state as functions of time are presented in Fig.~\ref{fig:asp}.

\begin{figure}[H]
    \centering
    \includegraphics[width=0.9\linewidth]{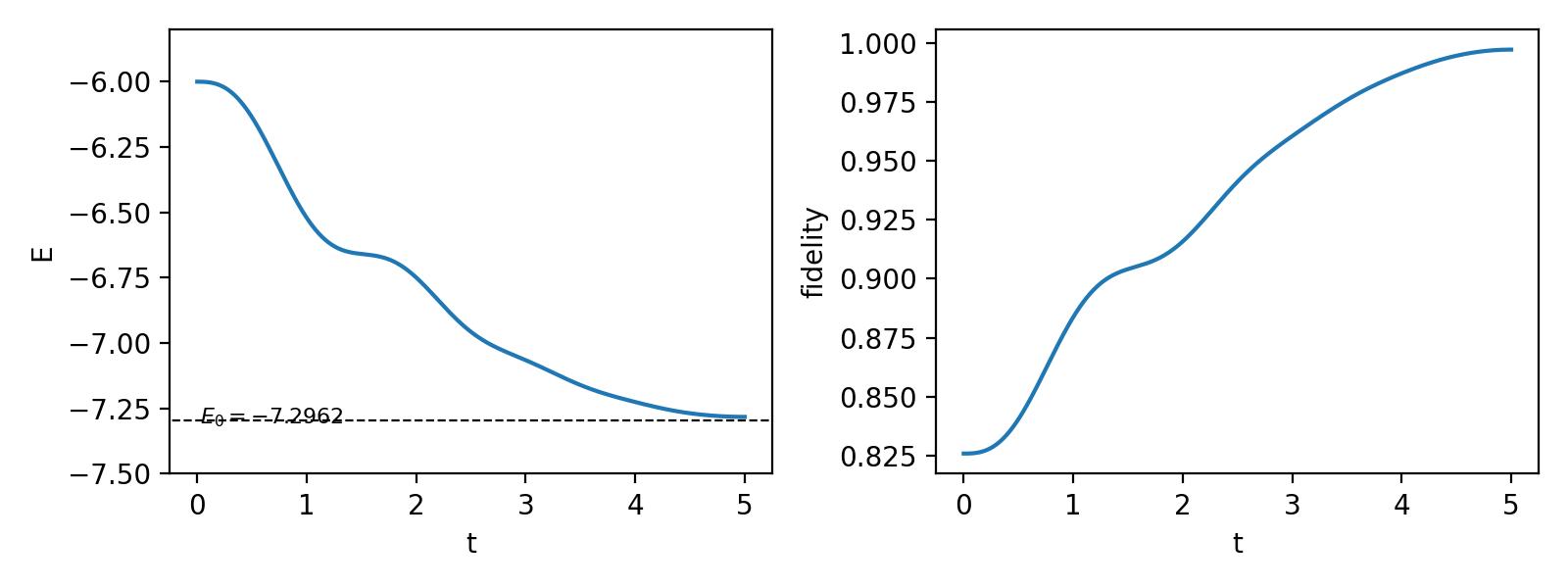}
    \caption{Numerical simulation result of adiabatic state preparation of the ground state of 1D transverse-field ising model via UAQS. The initial state is $|+\rangle^{\otimes 6}$. The left panel shows the evolution of the energy of the simulated state during the dynamics, while the right panel shows the evolution of the fidelity between the simulated state and the ground state. It can be seen that after the adiabatic evolution, the simulated quantum state has approached the target state.}
    \label{fig:asp}
\end{figure}

\subsection{Ground state preparation for molecular $H_2$}
The other application is ground state preparation via integrating imaginary-time evolution and UAQS. 
It is well-known that imaginary-time evolution provides a powerful and conceptually simple mechanism for preparing the ground state of a given Hamiltonian. By replacing real time $t$ with imaginary time $\tau = it$, the unitary Schrodinger evolution is transformed into a nonunitary process governed by
\begin{equation}
|\psi(\tau)\rangle = \frac{e^{-H \tau} |\psi(0)\rangle}{| e^{-H \tau} |\psi(0)\rangle |}.
\end{equation}
When the initial state has a nonzero overlap with the ground state of $H$, this evolution exponentially suppresses contributions from excited states according to their energy gaps, while amplifying the relative weight of the ground-state component. Concretely, any quantum state $|\psi\rangle$ can be expressed as a superposition of the orthogonal eigenstates of Hamiltonian $H$ of the system,
\begin{equation}
|\psi\rangle = \sum_j c_j|E_j\rangle,
\end{equation}
where $|E_j\rangle$ correspondes to the eigenstate with eigenvalue $E_j$. When the state $|\psi\rangle$ evolves in imaginary time evolution which governed by operator $e^{-H \tau}$, we have
\begin{equation}\label{eq:img}
|\psi_\tau\rangle = A_\tau e^{-H\tau}|\psi\rangle=A_\tau\sum_j c_je^{-E_j\tau}|E_j\rangle
\end{equation}
In Eq.~\eqref{eq:img}, each component of the state is multipied by the corresponding $e^{-Ej\tau}$. Since $E_0$ is the smallest eigenvalue, the term corresponding to the ground state $|E_0\rangle$ decayes the slowest as $\tau$ increases. Thus, when evolve the system in enough imaginary time, the higher energy components of state $|\psi\rangle$ are exponentially suppressed and leave the ground state as the dominant component. 
As a result, the state would quickly converges to the ground state of the Hamiltonian, up to normalization. This projection-like property makes imaginary-time evolution a natural and robust strategy for ground-state preparation, particularly in such frameworks where the nonunitary dynamics can be approximated through parameter optimization while remaining compatible with near-term quantum hardware.

Here, we apply such approach to prepare the ground state of Hydrogen molecule in the minimal STO-3G basis as an example. To approximate the ground state of $H_2$, we need to first transform the Hamiltonian of $H_2$ with Born-Oppenheimer approximation under second quantization,  and build a fermionic Hamiltonian in terms of creation and anihilation operators,
\begin{equation}
H = \sum_{i j} t_{i j} a_i^{\dagger} a_j+\frac{1}{2} \sum_{i j k l} V_{i j k l} a_i^{\dagger} a_j^{\dagger} a_l a_k
\end{equation}
where the $t_{i j}$ represent one-electron integrals which which account for the kinetic energy of the electrons and their interaction with the external potential,  and $V_{i j k l}$ describes the two-electron integrals, which correspond to the electron-electron Coulomb interaction. Then, we could perform Jordan-Wigner or Bravyi-Kitaev transformation to map the fermionic Hamiltonian which describes systems like electrons in atoms and molecules, to a qubit Hamiltonian that can be processed by quantum computers. Here, we also assume that the inital state has a non-zero overlap with the ground state. By using imaginary time evolution with parametrized analog anstaz, we are able to effectively project the initial state onto the ground state of the qubit Hamiltonian by suppressing contributions from excited states through exponential decay which allows for the accurate determination of the ground state of molecule $H_2$.
\begin{figure}[H]
    \centering
    \includegraphics[width=0.5\linewidth]{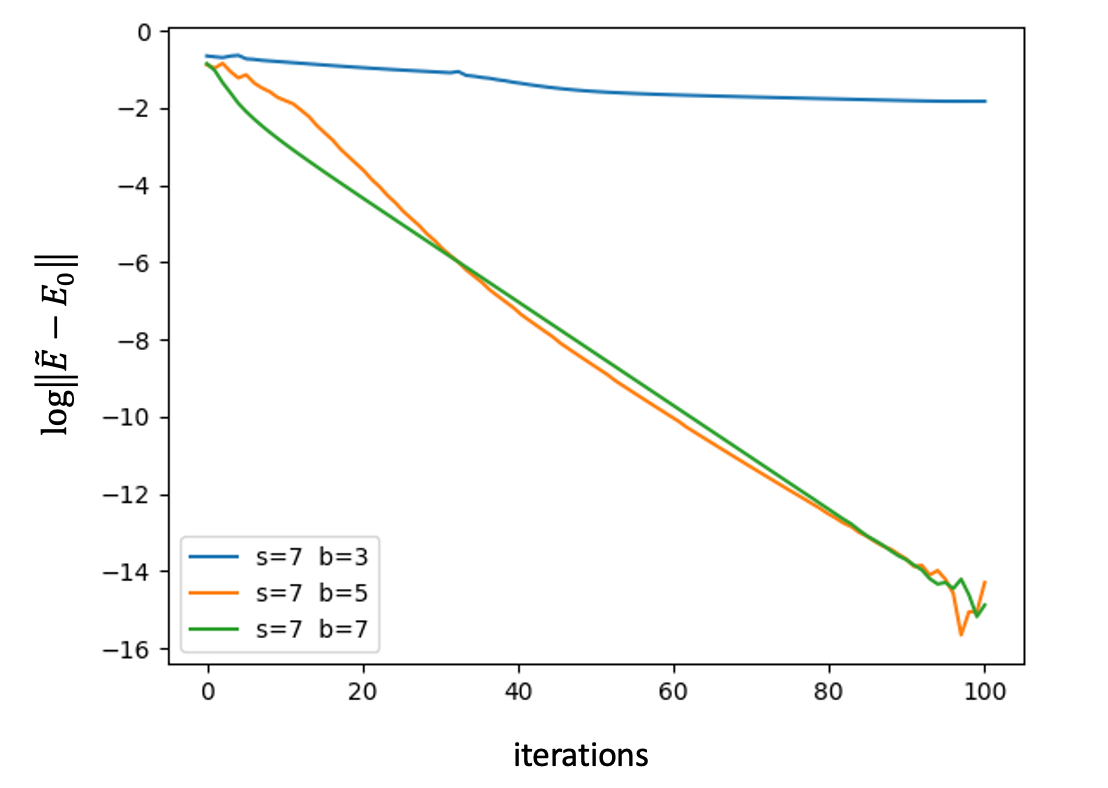}
    \caption{The performance of UAQS for estimating the ground energy of $H_2$. The total evolution time of UAQS is fixed to $s=7$, while $b$ denotes the number of independent control parameters assigned to each control Hamiltonian $H_j$. The y axis shows the logarithm of the error between the UAQS-estimated ground-state energy and the exact ground-state energy. Increasing $b$ can enhanced the expressivity of ansatz, which makes the convergence rate and final accuracy improved.}
    \label{fig:asp}
\end{figure}

\end{document}